\documentclass[journal]{IEEEtran}
\pdfoutput=1

\usepackage{amsthm}
\usepackage{graphicx}
\usepackage{epstopdf}
\usepackage{amsmath}
\usepackage{multicol}
\usepackage{stfloats}
\usepackage{amsfonts}
\usepackage{mathrsfs}
\usepackage{color}
\usepackage[ruled]{algorithm2e}
\usepackage{fancybox}
\usepackage{framed}
\usepackage{psfrag, amssymb, amsmath, cite}
\usepackage[colorlinks,linkcolor=red,anchorcolor=blue,citecolor=blue]{hyperref}
\newtheorem{theorem}{\textbf{Theorem}}
\newtheorem*{problem*}{\textbf{Problem}}

\newtheorem{definition}{\textbf{Definition}}
\newtheorem{lemma}{\textbf{Lemma}}

\newtheorem{remark}{\textbf{Remark}}

\renewcommand\footnoterule{\kern-3pt \hrule width 2in \kern 2.6pt}

\ifCLASSINFOpdf
\else
\fi
\begin{document}
%
\title{Cooperative event-based rigid formation control}
%
%
%

\author{Zhiyong~Sun,
        Qingchen~Liu,
        Na~Huang,
        Changbin~Yu,
        and~Brian~D.~O.~Anderson
\thanks{Z. Sun is with Department of Automatic Control, Lund University, Sweden. ({\tt\small \{zhiyong.sun@control.lth.se\})}
C. Yu is with  School of Engineering, Westlake University,
Hangzhou, China.
N. Huang and B. D. O. Anderson are with School of Automation, Hangzhou Dianzi University,
Hangzhou, China. Q. Liu and B. D. O. Anderson are with Data61-CSIRO and Research School
of Engineering, The Australian National University, Canberra ACT 2601, Australia.
  }}

%
%

\markboth{Journal of \LaTeX\ Class Files,~Vol.~xx, No.~x, December~2018}%
{Shell \MakeLowercase{\textit{et al.}}: Cooperative event-based formation control}
%



\maketitle

\begin{abstract}
This paper discusses cooperative stabilization control of rigid formations via an event-based approach. We first design a centralized event-based formation control system, in which a central event controller determines the next triggering time and broadcasts the event signal to all the agents for control input update. We then build on this approach to propose a distributed event control strategy, in which each agent can use its local event trigger and local information to update the control input at its own event time.  For both cases, the triggering  condition, event function and triggering behavior are discussed in detail, and the exponential convergence of the event-based formation system is  guaranteed.
\end{abstract}


%
\IEEEpeerreviewmaketitle

\section{Introduction}
\IEEEPARstart{F}{ormation} control of networked multi-agent systems has
received considerable attention in recent years due to its
extensive applications in many areas including both civil and military fields.
One problem of extensive interest  is \emph{formation shape
control}, i.e.   designing controllers to achieve or maintain a geometrical shape for the formation \cite{oh2012survey}. By using  graph rigidity theory, the formation
shape can be achieved by controlling a certain set of inter-agent
distances \cite{anderson2008rigid, krick2009stabilisation} and there is no requirement on a global coordinate system having to be known to all the agents.   This is in contrast to the linear \emph{consensus-based} formation control approach, in which the target formation is defined by a certain set of relative positions and a global coordinate system is required for all the agents to implement the  consensus-based formation control law (see detailed comparisons in \cite{oh2012survey}). Note that such a coordinate alignment condition is a rather strict requirement, which is undesirable for implementing formation controllers in e.g. a GPS-denied environment. Even if one assumes that such coordinate alignment is satisfied for all agents, slight coordinate misalignment, perhaps arising from sensor biases, can lead to a failure of formation control \cite{meng2016formation, meng20183}. Motivated by all these considerations, in this paper we focus on rigidity-based formation control.

 There have been rich works on controller design and stability analysis of rigid   formation control (see e.g. \cite{krick2009stabilisation, dorfler2010geometric, oh2014distance, anderson2014counting} and the review in \cite{oh2012survey}), most of which assume that the control input is updated in a continuous manner. The main objective of this paper is to provide alternative  controllers to stabilize rigid formation shapes based on an event-triggered approach.  This kind of controller design is attractive for real-world robots/vehicles equipped with \emph{digital} sensors or microprocessors \cite{astrom2008event, heemels2012introduction}. Furthermore,  by using an event-triggered mechanism to update the controller input, instead of using a continuous  updating strategy as discussed in e.g. \cite{krick2009stabilisation, dorfler2010geometric, oh2014distance, anderson2014counting}, the formation system can save resources in processors and thus can reduce  much of the computation/actuation burden for each agent.   Due to these favourable properties, event-based control has been studied extensively in recent years for linear and nonlinear systems \cite{sun2018event, ZhongkuiTNNLS,  zhang2018cooperative, Zhongkui2018TAC}, and especially for   networked control systems \cite{seyboth2012event, Duan2018, qin2018leader, jiang2018synchronization}. Examples of successful applications of the event-based control strategy on distributed control systems and networked control systems have been reported  in e.g. \cite{Qingchen2018, zhao2018edge,   duan2018event}.
 We refer the readers to the recent survey papers \cite{ge2018survey, zhang2017overview, PENG2018113, nowzari2017event} which provide  comprehensive and excellent reviews on event-based control for networked coordination and multi-agent systems.

Event-based control strategies in multi-agent formation systems have recently attracted increasing attention in the control community. Some recent attempts at applying event-triggered schemes in stabilizing multi-agent formations are reported in e.g.,  \cite{li2018event, ge2017distributed, viel2017distributed, yi2016formation}. However, these papers \cite{li2018event, ge2017distributed, viel2017distributed, yi2016formation} have focused on the event-triggered  consensus-based formation control approach, in which  the proposed event-based formation control schemes cannot be applied to stabilization control of rigid formation shapes.
We   note that   event-based rigid formation control has been discussed  briefly as an example of team-triggered network coordination in \cite{Cortes_team_triggered}. However, no thorough results have been reported in the literature to achieve cooperative rigid formation control with feasible and simple event-based solutions. This paper is a first contribution to advance the event-based control strategy to the design and implementation of rigid formation control systems.

Some preliminary results have been presented in  conference contributions \cite{sun2015event} and \cite{Liu2015event}. Compared to \cite{sun2015event} and \cite{Liu2015event}, this paper proposes new control methodologies to design event-based controllers as well as event-triggering functions. The event  controller presented in \cite{sun2015event} is a preliminary one, which only updates parts of the  control input and the control is not necessarily piecewise constant. This limitation is removed in the control design in this paper. Also, the complicated controllers in \cite{Liu2015event} have been simplified. Moreover, and centrally to the novelty of the paper, we will focus on the design of \emph{distributed }controllers based on novel event-triggering functions to achieve the cooperative formation task, while the event controllers in both \cite{sun2015event} and \cite{Liu2015event} are centralized. We prove that Zeno behavior is  excluded in the distributed event-based formation control system by   proving a positive lower bound of the inter-event triggering time. We also notice that a decentralized event-triggered control  was recently proposed in \cite{CCC_CHENFEI} to achieve rigid formation tracking control. The triggering strategy in \cite{CCC_CHENFEI} requires each agent to update the control input both at its own triggering time and its neighbors' triggering times, which increases the communication burden within the network. Furthermore, the triggering condition discussed in \cite{CCC_CHENFEI} adopts the position information in the event function design,  which   results in very complicated control functions and may limit their practical applications.

In this paper, we will provide a thorough study of rigid formation stabilization control with cooperative event-based approaches. To be specific, we will propose two feasible event-based control schemes (a centralized triggering scheme and a distributed triggering scheme) that aim to stabilize a general rigid formation shape. In this paper, by a careful design of the triggering condition and event function,
the aforementioned  communication requirements and controller complexity in e.g., \cite{sun2015event, CCC_CHENFEI, Liu2015event} are avoided.
For all cases, the triggering  condition, event function and triggering behavior are discussed in detail.
One of the key results in the controller performance analysis with both centralized and distributed event-based controllers is an exponential stability of the rigid formation system. The exponential stability of the formation control system has important consequences relating to \emph{robustness issues} in undirected rigid formations, as discussed  in the recent papers \cite{SMA14TACsub, sun2017rigid}.
%


The rest of this paper is organized as follows. In Section \ref{sec:pre}, preliminary concepts on graph theory and rigidity theory are introduced. In Section \ref{sec:centralized}, we propose a centralized event-based formation controller and discuss the convergence property and the exclusion of the Zeno behavior. Section \ref{sec:distributed} builds on the centralized scheme of Section \ref{sec:centralized} to develop a distributed event-based controller design, and presents detailed analysis of the triggering behavior and its feasibility. In Section \ref{sec:simulations}, some simulations are provided to demonstrate the controller performance. Finally, Section \ref{sec:conclusions} concludes this paper.

\section{Preliminaries} \label{sec:pre}
\subsection{Notations}
The notations used in this paper are fairly
standard. $\mathbb{R}^n$ denotes the $n$-dimensional Euclidean space. $\mathbb{R}^{m\times n}$
denotes the set of $m\times n$ real matrices. A matrix or vector transpose is denoted by a superscript $T$. The rank, image  and null space of a matrix $M$ are denoted by  $rank(M)$, $Im(M)$ and $ker(M)$, respectively.  The notation $\|M\|$ denotes the induced 2-norm of a matrix $M$ or the 2-norm of a vector $M$, and   $\|M\|_F$ denotes the Frobenius norm for a matrix $M$. Note that there holds $\|M\| \leq \|M\|_F$ for any matrix $M$ (see \cite[Chapter 5]{horn2012matrix}).
We use $\text{diag}\{x\}$ to denote a  diagonal matrix with the  vector $x$ on its diagonal, and $span\{v_1, v_2, \cdots, v_k\}$ to denote the subspace spanned by a set of vectors $v_1, v_2, \cdots, v_k$. The symbol $I_n$ denotes the $n \times n$ identity matrix, and  $\mathbf{1}_n$ denotes an $n$-tuple column vector of all ones.   The notation $\otimes$   represents the Kronecker product.

\subsection{Basic concepts on graph theory}
 Consider an undirected graph with $m$ edges and $n$ vertices, denoted by $\mathcal{G} =( \mathcal{V}, \mathcal{E})$  with vertex set $\mathcal{V} = \{1,2,\cdots, n\}$ and edge set $\mathcal{E} \subset \mathcal{V} \times \mathcal{V}$.  The neighbor set $\mathcal{N}_i$ of node $i$ is defined as $\mathcal{N}_i: = \{j \in \mathcal{V}: (i,j) \in \mathcal{E}\}$.
  The matrix relating the nodes to the edges is called the incidence matrix $H = \{h_{ij}\} \in \mathbb{R}^{m \times n}$, whose entries are defined as (with arbitrary edge orientations for the \emph{undirected} formations considered here)
     \begin{equation}
     h_{ij} =  \left\{
       \begin{array}{cc}
       1,  &\text{ the } i\text{-th edge sinks at node }j  \\ \nonumber
       -1,  &\text{ the } i\text{-th edge leaves  node }j  \\ \nonumber
       0,  & \text{otherwise}  \\
       \end{array}
      \right.
      \end{equation}
An important matrix representation of a graph $\mathcal{G}$ is the Laplacian matrix $L(\mathcal{G})$, which is defined as $L(\mathcal{G}) = H^TH$.  For a connected undirected graph, one has $rank(L) = n-1$ and $ker(L) = ker(H) = span\{\mathbf{1}_n\}$. Note that for the rigid formation modelled by an \emph{undirected} graph as considered in this paper, the orientation of each edge for writing the incidence matrix can be defined arbitrarily;  the graph Laplacian matrix  $L(\mathcal{G})$ for the undirected graph is always the same regardless of what edge orientations are chosen (i.e., is orientation-independent) and the stability analysis  remains unchanged.

\subsection{Basic concepts on rigidity theory}
Let $p_i  \in \mathbb{R}^d$ where $d = \{2,3\}$ denote a point that is assigned to $i \in \mathcal{V}$.  The stacked  vector $p=[p_1^T, \,
p_2^T, \cdots, \, p_n^T]^T \in \mathbb{R}^{dn}$ represents the realization of $\mathcal{G}$  in $\mathbb{R}^d$. The pair $(\mathcal{G}, p)$ is said to be a framework of $\mathcal{G}$ in  $\mathbb{R}^d$.  By introducing the matrix $\bar H: = H\otimes I_d \in \mathbb{R}^{dm \times dn}$, one can construct the relative position vector as an image of $\bar H$ from the position vector $p$:
\begin{equation}
z = \bar H p  \label{z_equation}
\end{equation}
where $z=[z_1^T, \,
z_2^T, \cdots, \, z_m^T]^T \in \mathbb{R}^{dm}$, with $z_k  \in \mathbb{R}^d$ being the relative position vector for the vertex pair defined by the $k$-th edge.

Using the same ordering of the edge set $\mathcal{E}$ as in  the definition of $H$, the rigidity function $r_{\mathcal{G}}(p): \mathbb{R}^{dn} \rightarrow \mathbb{R}^m$ associated with the framework $(\mathcal{G}, p)$ is given as:
\begin{equation}
r_{\mathcal{G}}(p) = \frac{1}{2} \left[\cdots, \|p_i-p_j\|^2,  \cdots \right]^T, \,\,\, (i,j) \in \mathcal{E}
\end{equation}
where the    $k$-th component in $r_{\mathcal{G}}(p)$, $\|p_i-p_j\|^2$, corresponds to the squared length of the relative position vector $z_k$ which connects the vertices $i$ and $j$.

The rigidity of frameworks is then defined as follows.
\begin{definition}  (see \cite{asimow1979rigidity}) A framework $(\mathcal{G}, p)$ is rigid in $\mathbb{R}^d$ if there exists a neighborhood $\mathbb{U}$ of $p$ such that $r_{\mathcal{G}}^{-1}(r_{\mathcal{G}}(p))\cap \mathbb{U} = r_{\mathcal{K}}^{-1}(r_{\mathcal{K}}(p))\cap \mathbb{U}$ where $\mathcal{K}$ is the complete graph with the same vertices as $\mathcal{G}$.
\end{definition}
In the following,  the set of all frameworks $(\mathcal{G}, p)$ which satisfies  the distance constraints is referred to as the set of \emph{target formations}. Let $d_{k_{ij}}$ denote the desired distance in the target formation which links agents $i$ and $j$. We further define \begin{align}
e_{k_{ij}} = \|p_i - p_j\|^2 - (d_{k_{ij}})^2
\end{align}
to denote the squared distance error for edge $k$.  Note that we will also use $e_k$ and $d_k$ occasionally  for notational convenience if no confusion is expected.  Define the distance square error vector $e = [e_1, \, e_2, \, \cdots, e_m]^T$.

One useful tool to characterize the rigidity property of a framework is the rigidity matrix $R \in \mathbb{R}^{m \times dn}$, which is defined as
\begin{equation}
R(p) = \frac{\partial r_{\mathcal{G}}(p) } {\partial p}  \label{R_matrix}
\end{equation}
It is not difficult to see that each row of the rigidity matrix $R$ takes the following form
\begin{equation}  \label{row_R}
[\mathbf{0}_{1 \times d}, \cdots, (p_i-p_j)^T, \cdots,\mathbf{0}_{1 \times d}, \cdots, (p_j-p_i)^T, \cdots, \mathbf{0}_{1 \times d}]
\end{equation}
Each edge gives rise to a row of $R$, and, if an edge links vertices $i$ and $j$, then the nonzero
entries of the corresponding row of $R$ are in the columns from $di - (d-1)$ to $di$ and from $dj - (d-1)$ to $dj$.
The equation \eqref{z_equation}  shows that the relative position vector lies in the image of $\bar H$. Thus one can  redefine the rigidity function, $g_{\mathcal{G}}(z): Im(\bar H) \rightarrow \mathbb{R}^m$ as $g_{\mathcal{G}}(z) = \frac{1}{2}\left[\|z_1\|^2, \|z_2\|^2,  \cdots, \|z_m\|^2 \right]^T$. From \eqref{z_equation} and \eqref{R_matrix}, one can obtain the following simple form for the rigidity matrix
\begin{eqnarray}
R(p) = \frac{\partial r_{\mathcal{G}}(p)}{\partial p} = \frac{\partial g_{\mathcal{G}}(z)}{\partial z} \frac{\partial z} {\partial p}  = Z^T \bar H  \label{rigidity_matrix}
\end{eqnarray}
where $Z = \textrm{diag} \{z_1,\,z_2,\cdots,\,z_m\}$.

The rigidity matrix will be used to determine the infinitesimal rigidity of the framework, as shown in the following definition. \\
\begin{definition} (see \cite{hendrickson1992conditions}) A framework   $(\mathcal{G}, p)$ is infinitesimally rigid in the $d$-dimensional space if
\begin{equation}
rank(R(p)) = dn-d(d+1)/2
\end{equation}
\end{definition}
Specifically, if the framework is infinitesimally rigid in $\mathbb{R}^2$ (resp. $\mathbb{R}^3$) and has exactly $2n-3$ (resp. $3n-6$) edges, then it is called a minimally and infinitesimally rigid framework. Fig. 1 shows several examples on rigid and nonrigid formations.
In this paper we focus on the stabilization problem of minimally and infinitesimally rigid formations.  {\footnote{With some complexity of calculation, the results
extend to non-minimally rigid formations (see \cite{SMA14TACsub, sun2016exponential}).}}
From the definition of infinitesimal rigidity, one can easily prove the following lemma:
\begin{lemma} \label{L_RP}
If the framework $(\mathcal{G}, p)$ is minimally and infinitesimally rigid in the $d$-dimensional space, then the matrix $R(p)R(p)^T$ is positive definite.
\end{lemma}

Another useful observation shows that there exists a smooth function which maps the distance set of a minimally rigid framework to the distance set of its corresponding framework modeled by a \emph{complete} graph.
\begin{lemma}\label{L_Smooth}

Let $r_{\mathcal{G}}(q)$ be the rigidity function for a given infinitesimally minimally rigid framework $(\mathcal{G}, q)$ with agents' position vector $q$. Further let $\bar r_{\bar {\mathcal{G}}}(q)$ denote the  rigidity function for an associated framework $(\bar {\mathcal{G}}, q)$, in which the vertex set remains the same as $(\mathcal{G}, q)$ but the underlying graph is a complete one (i.e. there exist $n(n-1)/2$ edges which link any  vertex pairs). Then there exists a continuously differentiable function $\pi: r_{\mathcal{G}}(q) \rightarrow  \mathbb{R}^{n(n-1)/2}$ for which $\bar r_{\bar {\mathcal{G}}}(q) = \pi(r_{\mathcal{G}}(q))$ holds locally.
\end{lemma}
Lemma \ref{L_Smooth} indicates that all the edge distances in the framework $(\bar {\mathcal{G}}, q)$ modeled by a complete graph can be expressed \textit{locally} in terms of the edge distances of a corresponding minimally  infinitesimally rigid framework $(\mathcal{G}, q)$ via  some smooth  functions.
The proof of the
above Lemma  is omitted here and can be found in  \cite{SMA14TACsub}.
We emphasize that Lemma~\ref{L_Smooth} is important for later analysis of a distance error system (a definition of the term will be given in Section \ref{subsec:central_design}). Lemma~\ref{L_Smooth} (together with Lemma \ref{L_smooth1} given later) will enable us to obtain a self-contained distance error system so that a Lyapunov argument  can be applied for convergence analysis.

\subsection{Problem statement}

The rigid formation control problem is formulated as follows.

\begin{problem*} Consider a group of $n$ agents in $d$-dimensional space modeled by single integrators
\begin{equation}
\dot p_i  = u_i,  \,\,\,\, i = 1,2,\cdots, n  \label{single_in}
\end{equation} Design a \emph{distributed} control $u_i \in \mathbb{R}^n$ for each agent $i$ in terms of $p_i-p_j$, $j \in \mathcal{N}_i$ with \emph{event-based} control update such that $||p_i-p_j||$ converges to the desired distance $d_{k_{ij}}$  which forms a  minimally and infinitesimally rigid formation.
\end{problem*}

In this paper, we will aim to propose two feasible event-based control strategies (a centralized triggering approach   and a distributed triggering approach) to solve this formation control problem.

\begin{figure}
  \centering
  \includegraphics[width=80mm]{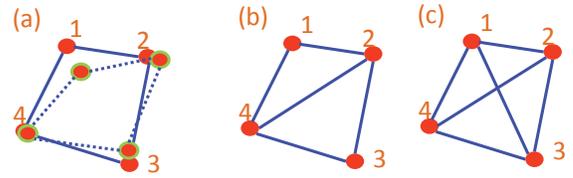}
  \caption{Examples on rigid and nonrigid formations. (a) non-rigid formation (a deformed formation with dashed lines is shown); (b) minimally rigid formation; (c) rigid but non-minimally rigid formation. }
  \label{shape_new}
  \label{fig:env}
\end{figure}

\section{Event-based controller design: Centralized case}  \label{sec:centralized}
This section focuses on the design of feasible event-based formation controllers, by assuming that a centralized processor is available for collecting the global information and broadcasting the triggering signal to all the agents such that their control inputs can be updated. The results in this section extend  the event-based formation control reported in \cite{sun2015event} by proposing an alternative approach for the event function design, which also simplifies the event-based controllers proposed in \cite{Liu2015event, CCC_CHENFEI}. Furthermore, the novel idea used for designing a   simpler event function in this section   will be useful for   designing  a feasible distributed version of an event-based formation control system, which will be reported in the next section.
\subsection{Centralized event controller design} \label{subsec:central_design}
We propose the following general form of  event-based formation control system
\begin{align}\label{position_system}
\dot p_i(t) &= u_i(t) = u_i(t_h)  \\ \nonumber
&= \sum_{j \in \mathcal{N}_i} (p_j(t_h) - p_i(t_h)) e_k(t_h)
\end{align}
for $t \in [t_h, t_{h+1})$, where  $h = 0,1,2, \cdots$ and $t_h$ is the $h$-th triggering time for updating new information in the controller. \footnote{  The differential equation \eqref{position_system} that models the event-based formation control system  involves switching controls at every event updating instant, for which  we understand its solution in the sense of Filippov \cite{cortes2008discontinuous}. } The control law is an obvious variant of the standard law for non-event-based formation shape control \cite{krick2009stabilisation, sun2016exponential}. Evidently, the control input takes  piecewise constant values in each time interval. In this section, we allow the switching times $t_h$ to be determined by a central controller. In a compact form, the above position system can be written as
\begin{equation} \label{position_system_compact}
\dot p(t) = - R(p(t_h))^T e(t_h)
\end{equation}

Denote a vector $\delta_i(t)$ as
\begin{align}
\delta_i(t) = & -\sum_{j \in \mathcal{N}_i} (p_j(t_h) - p_i(t_h)) e_k(t_h)  \nonumber \\
& + \sum_{j \in \mathcal{N}_i} (p_j(t) - p_i(t)) e_k(t)
\end{align}
 for $t \in [t_h, t_{h+1})$. Then the formation control system $\eqref{position_system}$ can be equivalently stated as
\begin{align}\label{position_system_new}
\dot p_i(t) =  u_i(t_h)  = \sum_{j \in \mathcal{N}_i} (p_j(t) - p_i(t)) e_k(t) - \delta_i(t)
\end{align}
Define a vector $\delta(t) = [\delta_1(t)^T, \delta_2(t)^T, \cdots, \delta_n(t)^T]^T \in \mathbb{R}^{dn}$. Then there holds
\begin{align}\label{eq:definition_delta}
\delta(t) =   R(t_h)^T e(t_h) - R(t)^T e(t)
\end{align}
which enables one to rewrite the compact form of the position system as
\begin{equation} \label{eq:new_position_system}
\dot p(t) = - R(t)^T e(t) -  \delta(t)
\end{equation}

To deal with the position system with the event-triggered controller \eqref{position_system}, we instead analyze   the distance error system. By noting that $\dot{e}(t)=2R(t)\dot{p}(t)$, the distance error system with the event-triggered controller \eqref{position_system} can be  derived as
\begin{align} \label{error_system}
\dot e(t) &=  2R(t) \dot p(t) \nonumber \\
&= -2R(t) R(p(t_h))^T e(t_h))\,\,\, \forall t \in [t_h, t_{h+1})
\end{align}
Note that all the entries of $R(t)$ and $e(t)$ contain  the real-time values of $p(t)$, and all the entries $R(p(t_h))$ and $e(t_h)$ contain  the piecewise-constant values $p(t_h)$ during the time interval $[t_h, t_{h+1})$.

The new form of the position system \eqref{eq:new_position_system} also implies that the compact form of the distance error system can be written as
\begin{equation} \label{error_system2}
\dot e(t) = -2R(t) (R(t)^T e(t)  +  \delta(t))
\end{equation}

Consider the function $V = \frac{1}{4}\sum_{k=1}^m e_k^2$ as a  Lyapunov-like  function candidate for \eqref{error_system2}. Similarly to the analysis in \cite{sun2015event}, we  define a sub-level  set $\mathcal{B}(\rho)= \{e: V(e) \leq \rho\}$ for some suitably small  $\rho$, such that when $e \in \mathcal{B}(\rho)$ the  formation is infinitesimally minimally rigid and $R(p(t))R(p(t))^T$ is positive definite.
Before giving the main proof, we record the following key result on the entries of the matrix $R(p(t))R(p(t))^T$.
\begin{lemma}\label{L_smooth1}
When the formation shape is close to the desired one such that the distance error $e$ is in the set $\mathcal{B}(\rho)$, the entries of the matrix $R(p(t))R(p(t))^T$ are continuously differentiable functions of $e$.
\end{lemma}
This lemma enables one to discuss the \emph{self-contained} distance error system \eqref{error_system2} and thus a Lyapunov argument can be applied to show the convergence of the  distance errors. The proof of Lemma \ref{L_smooth1} can be found in \cite{SMA14TACsub} or \cite{sun2014finite} and will not be presented here.  From Lemma \ref{L_smooth1},  one can show that
\begin{align} \label{derivative_lya2}
\dot V(t) &=    \frac{1}{2} e(t)^T  \dot e(t) =  -e(t)^T R(t) (R(t)^T e(t)  +  \delta(t))  \nonumber \\
&= -e(t)^T R(t)R(t)^T e(t) - e(t)^T R(t)  \delta(t) \nonumber \\
& \leq  -\|R(t)^T e(t)\|^2 + \|e(t)^T R(t)\|  \| \delta(t) \|
\end{align}

If we enforce the norm of $\delta(t)$   to satisfy
\begin{equation} \label{delta_norm}
\|\delta(t)\| \leq \gamma \|R(t)^T e(t)\|
\end{equation}
and choose the parameter  $\gamma$ to satisfy  $0< \gamma <1$,  then we can guarantee that
\begin{equation} \label{eq:Derivative_Lya}
\dot V(t) \leq (\gamma -1) \|R(t)^T e(t)\|^2 <0
\end{equation}
 This indicates that events are triggered when
\begin{equation} \label{trigger_function}
f  := \|\delta(t)\| -   \gamma \|R(t)^T e(t)\| = 0
\end{equation}
The event time $t_h$ is  defined to satisfy $f(t_h) = 0$ for $h = 0,1,2, \cdots$. For the time interval $t \in [t_h, t_{h+1})$, the control input is chosen as $u(t) = u(t_h)$  until the next event is triggered. Furthermore, every time when an event is triggered, the event vector $\delta$  will be reset to zero.

We also show two key properties of the formation control system \eqref{position_system} with the above event function \eqref{trigger_function}.
\begin{lemma} \label{lemma_fixed_centroid}
The formation centroid remains constant under the control of \eqref{position_system} with the event function   \eqref{trigger_function}.
\end{lemma}
\begin{proof}
Denote by $\bar p(t) \in \mathbb{R}^d$ the center of the mass of the formation, i.e., $\bar p(t) = \frac{1}{n} \sum_{i=1}^n p_i(t) = \frac{1}{n}(\mathbf{1}_n \otimes I_d)^T p(t)$.
One can show
\begin{align} \label{eq:formation_centroid_fixed}
\dot {\bar p}(t) &=   \frac{1}{n}(\mathbf{1}_n \otimes I_d)^T \dot p(t)  \nonumber \\
& = -\frac{1}{n}(\mathbf{1}_n \otimes I_d)^T R(p(t_h))^T e(t_h)  \nonumber \\
&=  -\frac{1}{n}\left( Z(t_h)^T \bar H (\mathbf{1}_n \otimes I_d) \right)^T e(t_h)
\end{align}
Note that  $ker(H) = span\{\mathbf{1}_n\}$ and therefore $ker(\bar H) = span\{\mathbf{1}_n \otimes I_d\}$. Thus $\dot {\bar p}(t) = 0$, which indicates that the formation centroid remains constant.
\end{proof}
The following lemma concerns the coordinate frame requirement and enables each agent to use its local coordinate frame to implement the control law, which is favourable for networked formation control systems in e.g. GPS denied environments.
\begin{lemma} \label{lemma:local_coordinate}
To implement the controller \eqref{position_system} with the event-based control update condition in \eqref{trigger_function}, each agent can use its own local coordinate frame which does not need to be aligned with a global coordinate frame.
\end{lemma}
The proof for the above lemma is omitted here, as it follows similar steps as in \cite[Lemma 4]{sun2014finite}. Note that Lemma  \ref{lemma:local_coordinate} implies the event-based formation system \eqref{position_system} guarantees the $SE(N)$ invariance of the controller, which is a nice property to enable convenient implementation  for networked control systems  without coordinate alignment for each individual agent \cite{SE_N_INVARIANCE}.

We now arrive at the following main result of this section.
\begin{theorem} \label{theorem_main_1}
Suppose the target formation is infinitesimally and minimally rigid and the initial formation shape is close to the target one. By using the above controller \eqref{position_system} and the event-triggering function \eqref{trigger_function}, all the agents will reach the desired formation shape locally exponentially fast.
\end{theorem}
\begin{proof}
The  above analysis relating to Eq. \eqref{derivative_lya2}-\eqref{trigger_function}  establishes boundedness of $e(t)$ since $\dot V$ is nonpositive. Now we show the exponential convergence of $e(t)$ to zero will occur from a ball around the origin, which is equivalent to the desired formation shape being reached
exponentially fast. According to Lemma  \ref{L_RP}, let $\bar \lambda_{\text{min}}$ denote the smallest eigenvalue of $M(e) :=R(p)R(p)^T$ when $e(p)$ is in the set $\mathcal{B} (\rho) $ (i.e. $\bar \lambda_{\text{min}}  =  \mathop {\min }\limits_{e \in \mathcal{B}(\rho)} \lambda  (M(e)) >0$).  Note that $\bar \lambda_{\text{min}}$ exists  because  the set $\mathcal{B} (\rho)$ is a compact set with respect to $e$ and   the eigenvalues of a matrix are continuous functions of the matrix elements. By recalling \eqref{eq:Derivative_Lya}, there further holds
\begin{align}
\dot V(t) \leq (\gamma - 1)  \bar \lambda_{\text{min}} \| e(t)\|^2
\end{align}
Thus one concludes
\begin{align} \label{eq:exponential_e}
\|e(t)\| \leq \text{exp}({- \kappa t}) \|e(0)\|
\end{align}
  with the exponential decaying rate  no less than $\kappa = 2(1- \gamma)  \bar \lambda_{\text{min}}$.
\end{proof}

Note that the convergence of the inter-agent distance error of itself does not directly guarantee the convergence of   agents' positions $p(t)$ to some fixed points, even though it does guarantee convergence to a correct formation shape. This is because that the desired equilibrium corresponding to the correct rigid shape is not a single point, but is a set of equilibrium points induced by rotational and translation invariance (for a detailed discussion to this subtle point, see \cite[Chapter 5]{dorfler_thesis}).
 A sufficient condition for this strong convergence to a stationary formation is guaranteed by the   exponential convergence, which was proved above.  To sum up, one has the following lemma on the convergence of the position system \eqref{position_system_compact} as a consequence of Theorem \ref{theorem_main_1}.
\begin{lemma} \label{lemma:fixed_final_position}
The event-triggered control law \eqref{position_system} and the event function \eqref{trigger_function} guarantee  the convergence of $p(t)$ to a fixed point.
\end{lemma}

\begin{remark}
We remark that the above Theorem 1 (as well as the subsequent results in later sections)  concerns a \emph{local} convergence. This is because   the rigid formation shape control system is \emph{nonlinear} and typically exhibits \emph{multiple} equilibria, which include the ones corresponding to correct formation shapes and those that do not correspond to   correct shapes.  It has been shown in \cite{anderson2014counting} by using the tool of Morse theory that multiple equilibria, including incorrect equilibria, are a consequence of any formation shape control algorithm which evolves in a steepest descent direction of a smooth cost function that is invariant under translations and rotations. A recent paper \cite{zhiyong_CDC15} proves the instability of a set of degenerate equilibria that live in a lower dimensional space. However, the stability property for more general equilibrium points is still unknown.  It is in fact considered as a very challenging open problem to obtain an almost global convergence result for general rigid formations, except for some special formation shapes such as 2-D triangular formation shape, or 2-D rectangular shape, or 3-D tetrahedral shape (see the review in \cite{oh2012survey, park2018distance}). We note that local convergence is still valuable in practice, if one assumes that initial shapes are close to the target ones (which is a very common assumption in most rigidity-based formation control works; see e.g.,  \cite{oh2012survey, krick2009stabilisation, dorfler2010geometric, ramazani2017rigidity, sun2014finite, sun2016exponential, Raik_SIAM}).
\end{remark}

\subsection{Exclusion of Zeno behavior}
In this section, we will analyze the exclusion of possible Zeno triggering behavior of the   event-based formation control  system \eqref{position_system}.
The following presents a formal definition of Zeno triggering (which is also termed \emph{Zeno execution} in the hybrid system study  \cite{ames2006stability}).
\begin{definition} \label{def:zeno-triggering}
For agent $i$, a triggering  is \emph{Zeno}  if
\begin{align}
\lim_{h \rightarrow \infty} t_h^i = \sum_{h=0}^{\infty} (t_{h+1}^i - t_h^i) = t_{\infty}^i
\end{align}
for some finite $t_{\infty}^i$ (termed the \emph{Zeno time}).
\end{definition}

Generally speaking, Zeno-triggering of an event controller means that it  triggers an infinite number of events  in a finite time period, which is an undesirable triggering behavior.
Therefore, it is important to exclude the possibility of Zeno behavior in an event-based system.
In the following we will   show that the event-triggered   system \eqref{position_system} does not exhibit Zeno behavior.

Note that the triggering function \eqref{trigger_function} involves the evolution of the term $R(t)^T e(t)$, whose derivative is calculated as

\begin{align}
\frac{\text{d} (R(t)^T e(t))}{\text{d}t} = & \dot R(t)^T e(t) + R(t)^T \dot e(t) \nonumber \\
= & \bar H^T  \dot Z(t)  e(t)   \nonumber \\
& -2 R(t)^T R(t) (R(t)^T e(t) +\delta(t))  \nonumber \\
\end{align}
According to the construction of the vector $\delta(t)$ in \eqref{eq:definition_delta}, there also holds $\dot \delta(t) = - \frac{\text{d} (R(t)^T e(t))}{\text{d}t}$.

Before presenting the proof, we first show a useful bound.
\begin{lemma} \label{lemma:bound_RTe}
The following bound  holds:
\begin{align} \label{eq:bound_RTe}
\|\bar H^T  \dot Z(t)  e(t)\|   &\leq \sqrt{d} \| \bar H^T \| \|\bar H \| \|e(t)\| \|\dot p(t)\| \nonumber \\
&\leq \sqrt{d} \| \bar H^T \| \|\bar H \| \|e(0)\| \|R(t)^T e(t) +\delta(t)\| \nonumber \\
\end{align}
\end{lemma}
\begin{proof}
We first show a trick to bound the term $\|  \dot Z(t)  e(t)\|$ by deriving an alternative expression for $\dot Z(t)  e(t)$:
\begin{small}
\begin{align}
\dot Z(t)  e(t) &= \text{diag}\{\dot z_1(t),   \dot z_2(t), \cdots,  \dot z_m(t)\} e(t) \nonumber \\
&=   \left[
\begin{array}{c}
e_1(t) \dot z_1(t) \\
e_2(t) \dot z_2(t) \\
\vdots \\
e_m(t) \dot z_m(t) \\
\end{array}
\right] \nonumber \\
& = \left( \left[
\begin{array}{cccc}
e_1(t)  &  0 &  \cdots &  0 \\
0 & e_2(t) & \cdots &  0    \\
\vdots         &   \vdots           &  \ddots &  \vdots    \\
0 & 0 &  \cdots & e_m(t)
\end{array}
\right] \otimes I_d \right)
 \left[
\begin{array}{c}
\dot z_1(t) \\
\dot z_2(t) \\
\vdots \\
\dot z_m(t) \\
\end{array}
\right] \nonumber \\
&  = : \left(E(t) \otimes I_d \right) \dot z(t) \nonumber \\
\end{align}
\end{small}
where $E(t)$ is defined as a diagonal matrix in the form $E(t) = \text{diag}\{e_1(t), e_2(t), \cdots, e_m(t)\}$.

Note that $z(t) = \bar H  p(t)$ and thus $\dot z(t) = \bar H \dot p(t)$. Then one has
\begin{align}
\| \bar H^T \dot Z(t)  e(t)\| &= \| \bar H^T \left(E(t) \otimes I_d \right) \dot z(t)\| \nonumber \\
&\leq \| \bar H^T \| \|\left(E(t) \otimes I_d \right)\| \|\bar H \dot p(t)\| \nonumber \\
&\leq \| \bar H^T \| \|\bar H \| \|\left(E(t) \otimes I_d \right)\|_F \|\dot p(t)\| \nonumber \\
&\leq \sqrt{d} \| \bar H^T \| \|\bar H \| \|e(t)\| \|R(t)^T e(t) +\delta(t)\|
\end{align}
where we have used the following facts
\begin{align}
&\|E(t)  \| \leq \|E(t)  \|_F \nonumber \\
& \|\left(E(t) \otimes I_d \right)\|_F = \sqrt{d}\|E(t)  \|_F \nonumber \\
& \|E(t)\|_F = \|e(t)\| \nonumber
 \end{align}
The first inequality in \eqref{eq:bound_RTe} is thus proved.
  The second inequality in \eqref{eq:bound_RTe} is due to the fact that $\|e(t)\| \leq \|e(0)\|, \forall t>0$ shown in \eqref{eq:exponential_e}.
\end{proof}

We now show that the Zeno triggering does not occur in the formation control system \eqref{position_system} with the triggering function  \eqref{trigger_function} by proving a positive lower bound on the inter-event time interval.
\begin{theorem}
The inter-event time interval $\{t_{h+1} - t_{h}\}$ is lower bounded by a positive value $\tau$
\begin{equation}
\tau =   \frac{\gamma}{\alpha(1+\gamma)} >0
\end{equation}
where

\begin{align} \label{eq:alpha}
\alpha &= \sqrt{d} \| \bar H^T \| \|\bar H \| \|e(0)\| +  \sqrt{2} \bar \lambda_{\text{max}}(R^TR(e))   >0 \nonumber \\
\end{align}
in which $\bar \lambda_{\text{max}}$ denotes the largest eigenvalue of $R^TR(e)$ when $e(p)$ is in the set $\mathcal{B} (\rho)$ (i.e. $\bar \lambda_{\text{max}}  =  \mathop {\text{max} }\limits_{e \in \mathcal{B} (\rho)} \lambda  (R^TR(e)) >0$), and
 $\gamma$ is the triggering parameter designed in \eqref{trigger_function} which satisfies  $\gamma \in (0,1)$.  Thus, Zeno triggering will not occur for the rigid formation control system  \eqref{position_system} with   the triggering function  \eqref{trigger_function}. 
\end{theorem}

\begin{proof}
%
%

We show that the growth of $\|\delta\|$ from 0 to the triggering threshold value $\gamma \|R^T e\|$ needs to take a positive time interval. To show this, the relative growth rate on $\|\delta(t)\|/ \|R(t)^T e(t)\|$ is considered. The following proof is  inspired by the one used in \cite{tabuada2007event}. In the following derivation, we omit the argument of time $t$ but it should be clear that each state variable and vector is considered as a function of $t$.

\begin{small}
\begin{align}
\frac{\text{d}}{\text{d}t} \frac{\|\delta\|}{\| R^T e\|}
\leq & \left (1+ \frac{\|\delta\|} {\| R^T e \|} \right)  \frac{ \|\dot R^T e + R^T \dot e\|}{\|R^T e\|} \nonumber \\
= &\left (1+ \frac{\|\delta\|} {\| R^T e \|} \right)  \frac{ \|\bar H^T \dot Z e + R^T \dot e\|}{\|R^T e\|} \nonumber \\
\leq &\left (1+ \frac{\|\delta\|} {\| R^T e \|} \right) \nonumber \\
& \left (    \frac{\sqrt{d} \| \bar H^T \| \|\bar H \| \|e\| \|\dot p\| + \|2R^T R(R^T e +\delta)\|}{\|R^T e\|} \right)  \nonumber \\
& \text{(appealing to Lemma \ref{lemma:bound_RTe})} \nonumber \\
\leq &\left (1+ \frac{\|\delta\|} {\| R^T e \|} \right) \nonumber \\
&\left ( (\sqrt{d} \| \bar H^T \| \|\bar H \| \|e(0)\| + \|2R^T R\|) \left (1 +\frac{\|\delta)\|}{\|R^T e\|}\right)  \right)  \nonumber \\
\leq  & \alpha\left (1+ \frac{\|\delta\|} {\| R^T e \|} \right)^2
\end{align}
\end{small}
where $\alpha$ is defined  in \eqref{eq:alpha}.   Note that   $\bar \lambda_{\text{max}}$ always exists and is finite (i.e. upper bounded) because  the set $\mathcal{B} (\rho)$ is a compact set with respect to $e$ and   the eigenvalues of a matrix are continuous functions of the matrix elements. Thus, $\alpha$   defined  in \eqref{eq:alpha} exists, which is   positive and upper bounded.   If we denote $\frac{\|\delta\|}{\| R^T e\|}$ by  $y$ we have the estimate $\dot y(t) \leq \alpha(1+y(t))^2$. By the comparison principle there holds $y(t) \leq \phi (t, \phi_0)$ where $\phi (t, \phi_0)$ is the solution of $\dot \phi = \alpha(1+ \phi)^2$ with initial condition $\phi(0, \phi_0) = \phi_0$.

Solving the differential equation for $\phi$ in the time interval $t \in [t_h, t_{h+1})$ yields $\phi(\tau, 0) = \frac{\tau \alpha}{1- \tau \alpha}$.
The inter-execution time interval is thus bounded by the time it takes for $\phi$ to evolve from 0 to $\gamma$. Solving the above equation, one obtains a positive lower bound   for the inter-event time interval $\tau =   \frac{\gamma}{\alpha(1+\gamma)}$. Thus, Zeno behavior is excluded for the formation control system \eqref{position_system}.
 The proof is completed.
\end{proof}
\begin{remark} We review several event-triggered formation strategies reported in the literature and highlight  the advantages of the event-triggered approach proposed in this section. In  \cite{sun2015event},  the triggering function is based on the information of the distance error $e$ only, which cannot guarantee a pure piecewise-constant update of the formation control input. The event function  designed in  \cite{Liu2015event} is based on the information of the relative position $z$, while the event function designed   \cite{CCC_CHENFEI} is   based on the  absolute position $p$. It is noted that   event functions and triggering conditions such as those in  \cite{Liu2015event} and \cite{CCC_CHENFEI} are very complicated, which may limit their practical applications.
The event function \eqref{trigger_function} designed in this section involves  the term $R^T e$, in which the  information of the relative position $z$ (involved in the entries of the rigidity matrix $R$) and of the distance error   $e$ has been included.
Such an event triggering  function    greatly reduces the controller complexity while at the same time also maintains the discrete-time update nature of the control input.
\end{remark}

%

\section{Event-based controller design: distributed case}  \label{sec:distributed}
\subsection{Distributed event controller design}

In this section we will further show how to design a distributed event-triggered formation controller in the sense that each agent can use only local measurements in terms of relative positions with respect to its neighbors to determine the next triggering time and control update value. Denote the event time for each agent $i$ as $t_0^i, t_1^i, \cdots, t_h^i, \cdots$. The dynamical system for agent $i$ to achieve the desired inter-agent distances is now rewritten as
\begin{align}\label{position_system_distributed_controller}
\dot p_i(t) &= u_i(t) = u_i(t_h^i), \forall t \in [t_h^i, t_{h+1}^i)
\end{align}
and we aim to design a distributed event function with feasible triggering condition such that the control input for agent $i$ is updated at its own event times $t_0^i, t_1^i, \cdots, t_h^i, \cdots$ based on local information. \footnote{ Again,   Filippov  solutions \cite{cortes2008discontinuous} are envisaged for the differential equation \eqref{position_system_distributed_controller}  with switching controls at every event updating instant. }

We consider the same Lyapunov function candidate as the one in Section \ref{sec:centralized}, but calculate the derivative as follows:
\begin{align} \label{derivative_lya_distributed}
\dot V(t) = &   \frac{1}{2} e(t)^T \dot e(t) =  -e(t)^T R(t) (R(t)^T e(t)  +  \delta(t))  \nonumber \\
= & -e(t)^T R(t)R(t)^T e(t) - e^T R(t)  \delta(t) \nonumber \\
 \leq & -\|R(t)^T e(t)\|^2 + \|e(t)^T R(t)  \delta(t) \|   \nonumber \\
\leq &- \sum_{i=1}^n \|\{R(t)^T e(t)\}_i\|^2 + \sum_{i=1}^n  \|\{R(t)^T e(t)\}_i\|  \| \delta_i(t) \|  \nonumber \\
\end{align}
where $\{R(t)^T e(t)\}_i \in \mathbb{R}^d$ is a vector block taken from the $(di-d+1)$th to the $(di)$th entries of the vector $R(t)^T e(t)$, and $\delta_i(t)$ is a   vector block taken from the $(di-d+1)$th to the $(di)$th entries of the vector $\delta(t)$. According to the definition of the rigidity matrix in \eqref{row_R}, it is obvious that $\{R(t)^T e(t)\}_i$ only involves local information  of agent $i$ in terms of relative position vectors $z_{k_{ij}}$ and distance errors $e_{k_{ij}}$ with $j \in \mathcal{N}_i$. Based on this, the control input for agent $i$ is designed as
\begin{align}\label{position_system_distributed}
\dot p_i(t) & =  u_i(t_h^i) =   \sum_{j \in \mathcal{N}_i} (p_j(t_h^i) - p_i(t_h^i)) e_k(t_h^i)) \nonumber \\
& = \{R(t_h^i)^T e(t_h^i)\}_i\,\,\,\forall t \in [t_h^i, t_{h+1}^i)
\end{align}

Note that there holds
\begin{align} \label{define_delta_i}
\delta_i(t) = \{R(t_h^i)^T e(t_h^i)\}_i - \{R(t)^T e(t)\}_i
\end{align}
and we can restate \eqref{position_system_distributed} as $\dot p_i(t) = - \{R(t)^T e(t)\}_i - \delta_i(t), t \in [t_h^i, t_{h+1}^i)$.

By using the inequality $\|\{R(t)^T e(t)\}_i\|  \| \delta(t)_i \| \leq \frac{1}{2a_i} \| \delta(t)_i \|^2 + \frac{a_i}{2} \|\{R(t)^T e(t)\}_i\|^2$ with $a_i \in (0,1)$, the above inequality \eqref{derivative_lya_distributed} on $\dot V$ can be further developed  as
\begin{align}
\dot V(t) \leq     & - \sum_{i=1}^n \|\{R(t)^T e(t)\}_i\|^2 \nonumber \\
 & + \sum_{i=1}^n  \frac{a_i}{2} \|\{R(t)^T e(t)\}_i\|^2 +  \sum_{i=1}^n   \frac{1}{2a_i} \| \delta_i(t) \|^2  \nonumber \\
= &  -\sum_{i=1}^n  \frac{2-a_i}{2}\|\{R(t)^T e(t)\}_i\|^2 +  \sum_{i=1}^n  \frac{1}{2a_i} \| \delta_i(t) \|^2    \nonumber
\end{align}

 If we enforce the norm of $\delta_i(t)$ to satisfy
\begin{equation} \label{distributed_trigger2}
    \frac{1}{2a_i} \| \delta_i(t) \|^2  \leq  \gamma_i \frac{2-a_i}{2}\|\{R(t)^T e(t)\}_i\|^2
\end{equation}
with $\gamma_i \in (0,1)$, we can guarantee
\begin{align} \label{derivative_lya_distributed_3}
\dot V(t) \leq \sum_{i=1}^n (\gamma_i -1) \frac{2-a_i}{2}\|\{R(t)^T e(t)\}_i\|^2
\end{align}

This implies that one can design a local triggering function for agent $i$  as
\begin{equation} \label{distributed_trigger_local_function}
f_i(t) : =  \| \delta_i(t) \|^2  - \gamma_i a_i (2-a_i)  \|\{R(t)^T e(t)\}_i\|^2
\end{equation}
and the event time $t_h^i$ for agent $i$ is defined to satisfy $f_i( t_h^i ) = 0$ for $h  = 0,1,2, \cdots$. For the time interval $t \in [t_h^i, t_{h+1}^i)$, the control input is chosen as $u_i(t) = u_i(t_h^i)$  until the next event for agent $i$ is triggered. Furthermore, every time   an event is triggered for agent $i$, the local event vector $\delta_i$  will be reset to zero.   Note that the condition $a_i \in (0,1)$ will also be justified in later analysis in Lemma \ref{lemma:feasibility}.

The convergence result with the distributed event-based formation controller and triggering function is summarized as follows.
\begin{theorem} \label{theorem:distributed_exponential}
Suppose that the target formation is infinitesimally and minimally rigid and the initial formation shape is close to the target one.  By using the above controller \eqref{position_system_distributed} and the distributed event  function \eqref{distributed_trigger_local_function}, all the agents will reach the desired formation shape locally exponentially fast.
\end{theorem}
\begin{proof}
The analysis is similar to Theorem \ref{theorem_main_1} and we omit several steps here. Based on the derivation in Eqs. \eqref{derivative_lya_distributed}-\eqref{distributed_trigger_local_function}, one can conclude that
\begin{align} \label{eq:exponential_e_distributed}
\|e(t)\| \leq \text{exp}(- \kappa t) \|e(0)\|
\end{align}
  with the exponential rate  no less than $\kappa = 2  \zeta_{\text{min}}  \bar \lambda_{\text{min}}$ where $\zeta_{\text{min}} = \text{min}_{i} (1- \gamma_i ) \frac{2-a_i}{2}$.
\end{proof}

The exponential convergence of $e(t)$ implies that the above local event-triggered controller \eqref{position_system_distributed} also guarantees the convergence of $p$ to a fixed point, by which one can conclude a similar result to the one in Lemma \ref{lemma:fixed_final_position}.

For the formation system with the distributed event-based controller, an analogous result to Lemma \ref{lemma:local_coordinate}  on coordinate frame requirement  is as follows.

\begin{lemma}
To implement the distributed formation controller \eqref{position_system_distributed}, each agent can use its own local coordinate frame to measure the relative positions to its neighbors and a global coordinate frame is not required. Furthermore, to detect the  distributed event condition \eqref{distributed_trigger2}, a local coordinate frame is sufficient which is not required to be aligned with the global coordinate frame.
\end{lemma}

\begin{proof}
The proof of the first statement on the distributed controller \eqref{position_system_distributed} follows similar steps as in \cite[Lemma 4]{sun2014finite} and is omitted here. We then prove the second statement on the event condition   \eqref{distributed_trigger2}. Suppose agent $i$'s position in a global coordinate frame is measured as $p_i^g$, while $p_i^i$ and $p_j^i$ stand  for agent $i$ and its neighboring agent $j$'s positions measured by agent $i$'s  local coordinate frame.  Clearly,  there exist a rotation matrix $\mathcal{Q}_i \in \mathbb{R}^{d \times d}$ and a translation vector $\vartheta_i \in \mathbb{R}^{d}$, such that $p_j^i = \mathcal{Q}_i p_j^g + \vartheta_i$. We also denote the relative position   between agent $i$ and agent $j$ as $z_{k_{ij}}^i$  measured by agent $i$'s local coordinate frame, and $z_{k_{ij}}^g$   measured by the global coordinate frame. Obviously there holds $z_{k_{ij}}^i = p_j^i -  p_i^i = \mathcal{Q}_i ( p_j^g -  p_i^g) = \mathcal{Q}_i z_{k_{ij}}^g$ and thus  $\|z_{k_{ij}}^i\| = \|z_{k_{ij}}^g\|$.  Also notice that the event condition   \eqref{distributed_trigger2} involves the terms $\delta_i$ and $\{R(t)^T e(t)\}_i$ which are functions of the relative position vector $z$, and thus event detection using \eqref{distributed_trigger_local_function} remains unchanged regardless of what coordinate frames are used.   Since $\mathcal{Q}_i$ and $\vartheta_i$  are chosen arbitrarily, the above analysis concludes  that the detection of the local event condition \eqref{distributed_trigger_local_function} is independent of a global coordinate basis, which implies that agent $i$'s local coordinate frame is sufficient to implement \eqref{distributed_trigger2}.
\end{proof}
The above lemma indicates that  the distributed event-based controller \eqref{position_system_distributed} and distributed event  function \eqref{distributed_trigger_local_function} still guarantee  the $SE(N)$ invariance property and enable a convenient implementation  for the proposed formation control system  without coordinate alignment for each individual agent.

Differently to Lemma \ref{lemma_fixed_centroid}, we show that the distributed event-based controller proposed in this section cannot guarantee a fixed formation centroid.

\begin{lemma}
The position of the formation centroid is not guaranteed to be fixed when the distributed event-based controller \eqref{position_system_distributed} and event function \eqref{distributed_trigger_local_function} are applied.
\end{lemma}
\begin{proof}
The dynamics for the formation centroid can be derived as
\begin{align} \nonumber
\dot {\bar p}(t) &=   \frac{1}{n}(\mathbf{1}_n \otimes I_d)^T \dot p(t)
\end{align}
However, due to the \emph{asymmetric} update of each agent's control input by using the local event function \eqref{distributed_trigger_local_function} to determine a local triggering time, one cannot decompose the vector $\dot p(t)$ into   terms involving $\bar H$ and a single distance error vector as in \eqref{eq:formation_centroid_fixed}. Thus $\dot {\bar p}(t)$ is not guaranteed to be zero and there exist  motions for the formation centroid when the distributed event-based controller \eqref{position_system_distributed} is applied.
\end{proof}

\begin{remark}
We note a key property of the distributed event-based controller  \eqref{position_system_distributed} and \eqref{distributed_trigger_local_function} proposed in this section. It is obvious from \eqref{position_system_distributed} and \eqref{distributed_trigger_local_function} that each agent $i$ updates  its own control input by using only local information in terms of relative positions of its neighbors (which can be measured by agent $i$'s local coordinate system), and is not affected by the control input updates from its neighbors. Thus, such local event-triggered approach does not require any communication between any two agents.
\end{remark}

\subsection{Triggering behavior analysis}
This subsection aims to analyze some properties of the distributed event-triggered control strategy proposed above.  Generally speaking, singular triggering means   no more triggering exists after a feasible triggering event, and continuous triggering means that  events are triggered continuously.  For the definitions of singular triggering and continuous triggering, we refer the reader to \cite{fan2013distributed}. The following lemma shows the triggering feasibility with the local and distributed event-based controller \eqref{position_system_distributed}  by excluding these two cases.

\begin{lemma} \label{lemma:feasibility}
(Triggering feasibility) Consider the distributed formation system with the distributed event-based controller \eqref{position_system_distributed} and  the distributed event function \eqref{distributed_trigger_local_function}. If there exists $t_h^i$ such that $\{R(t_h^i)^T e(t_h^i)\}_i \neq 0$, then
\begin{itemize}
\item (No singular triggering) agent $i$ will not exhibit singular triggering for all $t>t_h^i$.
\item (No continuous triggering) agent $i$ will not exhibit continuous triggering for all $t>t_h^i$.
\end{itemize}
\end{lemma}

\begin{proof}
The proof is inspired by \cite{fan2013distributed}. First note that due to $a_i \in (0,1)$, there holds $a_i(2-a_i)  \in (0,1)$ and because $\gamma_i \in (0,1)$, there further holds $\gamma_i a_i(2-a_i) \in (0,1)$. For notational convenience, define
\begin{align} \label{eq:varrho}
\varrho_i: = \gamma_i a_i(2-a_i)
\end{align}
 From \eqref{define_delta_i} and \eqref{distributed_trigger2}, for $t \in [t_h^i, t_{h+1}^i)$ one has
\begin{align}
\underline{\chi}_i := \frac{\|\{R(t_h^i)^T e(t_h^i)\}_i\|}{1+\sqrt{\varrho_i}} & \leq \|\{R(t)^T e(t)\}_i\|  \nonumber \\
& \leq  \frac{\|\{R(t_h^i)^T e(t_h^i)\}_i\|}{1-\sqrt{\varrho_i}} =: \overline{\chi}_i
\end{align}
We first prove the   statement on non-singular triggering. According to the definition of the event-triggering function \eqref{distributed_trigger_local_function}, local events for agent $i$ can only occur either when  $\|\{R(t)^T e(t)\}_i\|$ equals $\underline{\chi}_i$ or when $\|\{R(t)^T e(t)\}_i\|$ equals $\overline{\chi}_i$. Note that $\|\{R(t)^T e(t)\}_i\|^2 \leq \|\{R(t)^T e(t)\}\|^2 \leq \bar \lambda_{\text{max}}(R(t)^TR(t))  \|e(t)\|^2$, where $\bar \lambda_{\text{max}}$ is the largest eigenvalue of $R(t)^TR(t)$ which is bounded for any  $e \in \mathcal{B}(\rho)$. Also note that $\|e(t)\|^2$ decays exponentially fast to zero as proved in Theorem 3. This implies that $\|\{R(t)^T e(t)\}_i\|$ will eventually decrease to $\underline{\chi}_i$.  By assuming that $\{R(t_h^i)^T e(t_h^i)\}_i \neq 0$, the next event time $t_{h+1}^i$ for agent $i$ always exists with $\{R(t_{h+1}^i)^T e(t_{h+1}^i)\}_i \neq 0$.

The second statement can be proved by using similar arguments to those above, and by observing that $\{R(t)^T e(t)\}_i$ evolves  continuously and a local event is triggered if and only if \eqref{distributed_trigger_local_function} is satisfied.
\end{proof}
 In the following we will further discuss the possibility of the Zeno behavior in the distributed event-based formation system \eqref{position_system_distributed}.

\begin{theorem}  \label{theorem:no_zeno1}
(Exclusion of Zeno triggering) Consider the distributed formation system with the distributed event-based controller \eqref{position_system_distributed} and  the distributed event function \eqref{distributed_trigger_local_function}.
\begin{itemize}
\item At least one agent does not exhibit Zeno triggering behavior.
\item In addition, if there exists $\epsilon >0$ such that $\|\{R(t)^T e(t)\}_i\|^2 \geq \epsilon \|e(t)\|^2$ for all $i =1,2, \cdots, n$ and $t \geq 0$, then there exists a common positive lower bound for any inter-event time interval for each agent. In this case, no agent will exhibit Zeno triggering behavior.
\end{itemize}
\end{theorem}

\begin{proof}
Note that $\| \delta_i(t) \|^2 \leq \| \delta(t) \|^2$ for any $i$. In addition, there exists an agent $i_*$ such that $\|\{R(t)^T e(t)\}_{i_*}\|^2 \geq \frac{1}{m} \|\{R(t)^T e(t)\}\|^2$.
Then one has
\begin{align}
\frac{\| \delta_{i_*}(t) \|}{\|\{R(t)^T e(t)\}_{i_*}\|} \leq \sqrt{m} \frac{\| \delta(t) \|}{\|\{R(t)^T e(t)\}\|}
\end{align}

By recalling the proof in   Theorem 2, we can conclude that the inter-event interval for agent $i_*$ is bounded from below by a time $\tau_{i_*}$ that satisfies
\begin{align}
\sqrt{m} \frac{\tau_{i_*} \alpha}{1- \tau_{i_*} \alpha} = \sqrt{\gamma_{i_*} a_{i_*}(2-a_{i_*})}
\end{align}
So that $\tau_{i_*} = \frac{\sqrt{\gamma_{i_*} a_{i_*}(2-a_{i_*})}}{ \alpha (\sqrt{m}+\sqrt{\gamma_{i_*} a_{i_*}(2-a_{i_*})})} >0$. The first statement is proved.

We then prove the second statement.  Denote $\bar \lambda_{\text{max}}$ as the maximum of $\lambda_{\text{max}} (R^T R(e))$ for all $e \in \mathcal{B}(\rho)$. Since $\mathcal{B}(\rho)$ is a compact set, $\bar \lambda_{\text{max}}$ exists and is bounded. Then there holds $\|\{R(t)^T e(t)\}\|^2 \leq  \bar \lambda_{\text{max}}  \|e(t)\|^2$.  One can further show
\begin{align}
\|\{R(t)^T e(t)\}_i\|^2 \geq \epsilon \|e(t)\|^2 \geq    \frac{\epsilon}{\bar \lambda_{\text{max}}} \|\{R(t)^T e(t)\}\|^2
\end{align}
By following a similar argument to that  above and using the analysis in the proof of Theorem 2, a  lower bound on the inter-event interval $\bar \tau_i$ for each agent can be calculated as
\begin{align}
\bar \tau_i = \frac{\sqrt{\varrho_i}}{  \alpha \left(\sqrt{\frac{ \bar \lambda_{\text{max}}}{\epsilon}}+\sqrt{\varrho_i}\right)} >0, \,\,i = 1,2, \cdots, n
\end{align}
The proof is completed.
\end{proof}
\begin{remark}
 The first part of Theorem \ref{theorem:no_zeno1} is motivated by \cite[Theorem 4]{dimarogonas2012distributed}, which guarantees the exclusion of Zeno behavior for at least one agent. To improve the result for all the agents, we propose a condition in the second part of  Theorem \ref{theorem:no_zeno1}.
The above results in Theorem \ref{theorem:no_zeno1} on the distributed event-based controller are more conservative than the centralized case. The condition on the existence of $\epsilon >0$   essentially guarantees that $\|\{R(t)^T e(t)\}_i\|$ cannot be zero at any finite time, and will be zero if and only if $t = \infty$.   By a similar analysis from event-based multi-agent consensus dynamics in    \cite{SUN2018264}, one can show that if $\|\{R(t)^T e(\hat t)\}_i\| = 0$ at some finite time instant $\hat t$, then   agent $i$ will exhibit a Zeno triggering and  the time $\hat t$ is a Zeno time for agent $i$.
However, we have performed many simulations with different rigid formation shapes and observed that in most cases $\|\{R(t)^T e(t)\}_i\|$ is non-zero. We conjecture  that this may be due to the infinitesimal rigidity  of the formation shape. In the next subsection however,   we will provide a simple modification of the distributed controller to remove the condition on $\epsilon$.
\end{remark}
\subsection{A modified distributed event function}
The event function \eqref{distributed_trigger_local_function} for agent $i$ involves the comparison of two terms, i.e. $\| \delta_i(t) \|^2$ and $\gamma_i a_i (2-a_i)  \|\{R(t)^T e(t)\}_i\|^2$. As noted above,   the existence of $\epsilon >0$ can guarantee $\|\{R(t)^T e(t)\}_i\| \neq 0$   for any finite time $t$. To remove this condition in Theorem \ref{theorem:no_zeno1}, and motivated by \cite{sun2018event}, we propose the following modified event function by including an exponential decay term:
\begin{align} \label{distributed_modified_event}
f_i(t) : =  \| \delta_i(t) \|^2 &  -   \gamma_i a_i (2-a_i)  \|\{R(t)^T e(t)\}_i\|^2  \nonumber \\
 &  - 2a_i v_i \text{exp}(-\theta_i t)
\end{align}
where $v_i>0, \theta_i>0$ are  parameters that can be adjusted in the design to control the formation convergence speed.   Note that $v_i \text{exp}(-\theta_i t)$ is always positive and converges to zero when $t \rightarrow \infty$. Thus, even if $\{R(t)^T e(t)\}_i$ exhibits a crossing-zero scenario at some finite time instant,   the addition of this decay term guarantees a positive threshold value in the event function which  avoids the case of comparing  $\| \delta_i(t) \|^2$ to a zero threshold.

The main result in this subsection is to show that the above modified event function ensures Zeno-free triggering for all   agents, and also drives the formation shape to reach the target one.
\begin{theorem}
By using the proposed distributed event-based formation controller \eqref{position_system_distributed} and the modified distributed event  function \eqref{distributed_modified_event}, all the agents will reach the desired formation shape locally exponentially fast and no agent will exhibit Zeno-triggering behavior.
\end{theorem}
\begin{proof}

We consider the same Lyapunov function as used in Theorem 1 and Theorem 3 and follow similar steps as above. The triggering condition from the modified event function  \eqref{distributed_modified_event} yields
\begin{align} \label{derivative_lya_distributed_modified}
\dot V(t) \leq \sum_{i=1}^n (\gamma_i -1) \frac{2-a_i}{2}\|\{R(t)^T e(t)\}_i\|^2 + v_i \text{exp}(-\theta_i t)
\end{align}
which follows that
\begin{align}
\dot V(t) \leq - 4\zeta_{\text{min}} \bar \lambda_{\text{min}} V(t) + \sum_{i=1}^n v_i \text{exp}(-\theta_i t)
\end{align}
where $\zeta_{\text{min}}$ is defined as the same to the notation in Theorem 3 (i.e. $\zeta_{\text{min}} = \text{min}_{i} (1- \gamma_i ) \frac{2-a_i}{2}$). For notational convenience, we define $\kappa = 2  \zeta_{\text{min}}  \bar \lambda_{\text{min}}$ (also the same to Theorem 3).
By the well-known comparison principle \cite[Chapter 3.4]{khalil1996nonlinear}, it further follows that
\begin{align}
V(t) \leq \,\, & \text{exp}(- 2 \kappa t) V(0) \nonumber \\
& + \sum_{i=1}^n \frac{v_i}{2 \kappa - \theta_i}  (\text{exp}(-\theta_i t) - \text{exp}(-2 \kappa t))
\end{align}
which implies that $V(t) \rightarrow 0$ as $t \rightarrow \infty$, or equivalently, $\|e(t)\| \rightarrow 0$, as $t   \rightarrow \infty$.

In the following analysis showing exclusion of   Zeno behavior we let $t \in [t_h^i, t_{h+1}^i)$.
We first show a sufficient condition to guarantee $f_i(t) \leq 0$ when $f_i(t)$ is defined in \eqref{distributed_modified_event}. Note that $f_i(t) \leq 0$  can be equivalently stated as
\begin{align} \label{eq:inequality_event_condition}
&(\varrho_i  +1 ) \| \delta_i(t) \|^2  \nonumber \\
&\leq  \varrho_i  ( \| \delta_i(t) \|^2 + \|\{R(t)^T e(t)\}_i\|^2) +  2a_i v_i \text{exp}(-\theta_i t)
\end{align}
where $\varrho_i$ is defined in \eqref{eq:varrho}. Note that
\begin{align}
\|\{R(t_h^i)^T e(t_h^i)\}_i\|^2 = &\| \delta_i(t)  + \{R(t)^T e(t)\}_i\|^2   \nonumber \\
\leq  & 2 (\| \delta_i(t) \|^2 + \|\{R(t)^T e(t)\}_i\|^2)
\end{align}
Thus, a sufficient condition to guarantee the above inequality \eqref{eq:inequality_event_condition} (and the inequality $f_i(t) \leq 0$) is
\begin{align} \label{eq:sufficient_condition_event}
 \| \delta_i(t) \|^2   \leq  \frac{\varrho_i}{(2\varrho_i  +2 )}  \|\{R(t_h^i)^T e(t_h^i)\}_i\|^2
  + \frac{2a_i v_i}{\varrho_i  +1}  \text{exp}(-\theta_i t)
\end{align}

Note that from \eqref{eq:definition_delta} there holds $\dot \delta_i = - \frac{\text{d}}{\text{d}t} \{R(t)^T e(t)\}_i$. It follows that
\begin{small}
\begin{align}
\frac{\text{d}}{\text{d}t} \|\delta_i(t)\| & \leq \frac{\|\delta_i(t)^T\|}{\|\delta_i(t)\|} \|\dot \delta_i(t)\| \nonumber \\
&  = \|\frac{\text{d}}{\text{d}t} \{R(t)^T e(t)\}_i \|  \nonumber \\
& = \|\sum_{j \in \mathcal{N}_i} (\dot p_j(t) - \dot p_i(t)) e_k(t) +  \sum_{j \in \mathcal{N}_i} ( p_j(t) -  p_i(t)) \dot e_k(t)\| \nonumber \\
& = \|\sum_{j \in \mathcal{N}_i} (e_{k_{ij}}(t) \otimes I_d + 2 z_{k_{ij}}(t) z_{k_{ij}}(t)^T) (\dot p_j(t) - \dot p_i(t)) \| \nonumber \\
& = \|\sum_{j \in \mathcal{N}_i} Q_{ij}(t) (\{R(t_{h'}^j)^T e(t_{h'}^j)\}_j - \{R(t_h^i)^T e(t_h^i)\}_i) \| \nonumber \\
& \leq \sum_{j \in \mathcal{N}_i} \|Q_{ij}(t) \| \|(\{R(t_{h'}^j)^T e(t_{h'}^j)\}_j - \{R(t_h^i)^T e(t_h^i)\}_i) \|  \nonumber \\
& := \alpha_i
\end{align}
\end{small}
where $Q_{ij}(t) : = e_{k_{ij}}(t) \otimes I_d + 2 z_{k_{ij}}(t) z_{k_{ij}}(t)^T$, and $t_{h'}^j = \text{arg max}_{h} \{t_h^j | t_h^j \leq t, j \in \mathcal{N}_i \}$. By a straightforward argument similar to Lemma  \ref{lemma:feasibility}, it can be shown that agent $i$   will not exhibit singular
triggering, which indicates $\alpha_i$ cannot be zero for all time intervals.  Also note that $\alpha_i$ is upper bounded by a positive constant which implies that $\frac{\text{d}}{\text{d}t} \|\delta_i(t)\|$ is upper bounded. From the sufficient condition given in \eqref{eq:sufficient_condition_event} which guarantees the event triggering condition $f_i(t) \leq 0$ shown in \eqref{distributed_modified_event}, the next inter-event interval for agent $i$ is lower bounded by the solution $\tau_{h}^i$ of the following equation

\begin{small}
\begin{align} \label{eq:Zeno_solution}
\tau_{h}^i \alpha_i = \sqrt{\frac{\varrho_i}{(2\varrho_i  +2 )}  \|\{R(t_h^i)^T e(t_h^i)\}_i\|^2
  + \frac{2a_i v_i}{\varrho_i  +1}  \text{exp}(-\theta_i (t_h^i + \tau_{h}^i))}
\end{align}
\end{small}

We note the solution $\tau_{h}^i$ to \eqref{eq:Zeno_solution} always exists and is positive. Thus, no agents will exhibit Zeno-triggering behavior with the modified event function \eqref{distributed_modified_event}.
\end{proof}

\begin{remark}
In the  modified  event  function \eqref{distributed_modified_event}, a positive and exponential decay term $v_i \text{exp}(-\theta_i t)$ is included which guarantees that, even if $\{R(t)^T e(t)\}_i$ becomes zero at some finite time, the inter-event time interval at any finite time is positive and thus Zeno triggering is excluded. A more general strategy for designing a Zeno-free event  function is to include a positive $\mathcal{L}^p$ signal in the event function \eqref{distributed_modified_event}. By following a similar analysis to \cite{sun2018event}, one can show that, if the term $v_i \text{exp}(-\theta_i t)$ in \eqref{distributed_modified_event} is replaced by a general positive $\mathcal{L}^p$ signal, the local convergence of the event-based formation control system \eqref{position_system_distributed} with Zeno-free triggering for all agents still holds.
However, a general $\mathcal{L}^p$ signal in \eqref{distributed_modified_event} may only guarantee an asymptotic convergence rather than an exponential convergence of the overall formation system \eqref{position_system_distributed}.
\end{remark}


%
%

\section{Simulation results}  \label{sec:simulations}
In this section we provide three sets of simulations to show the behavior of certain formations with a centralized event-based controller and a distributed event-based controller, respectively.  Consider a double tetrahedron formation in $\mathbb{R}^3$, with the desired distances for each edge being $2$. The initial conditions for each agent are chosen as $p_1(0) = [0,-1.0,0.5]^T$, $p_2(0) = [1.8,1.6,-0.1]^T$, $p_3(0) = [-0.2,1.8,0.05]^T$, $p_4(0) = [1.2,1.9,1.7]^T$ and $p_5(0) = [-1.0,-1.5,-1.2]^T$, so that the initial formation is close to the target one. The parameter $\gamma$ in the trigger function is set as $\gamma = 0.6$. Figs.~2-4   illustrate formation convergence and event performance with a centralized event triggering function. The trajectories of each agent, together with the initial shape and final shape are depicted in Fig. 2. The trajectories of each distance error are depicted in Fig. 3, which shows an exponential convergence to the origin. Fig. 4 shows the triggering time instant and the evolution of the norm of the vector $\delta$ in the triggering function \eqref{trigger_function}, which is obviously bounded below by $\gamma \|R(t)^T e(t)\|$ as required by \eqref{delta_norm}.

\begin{figure}
  \centering
  \includegraphics[width=75mm]{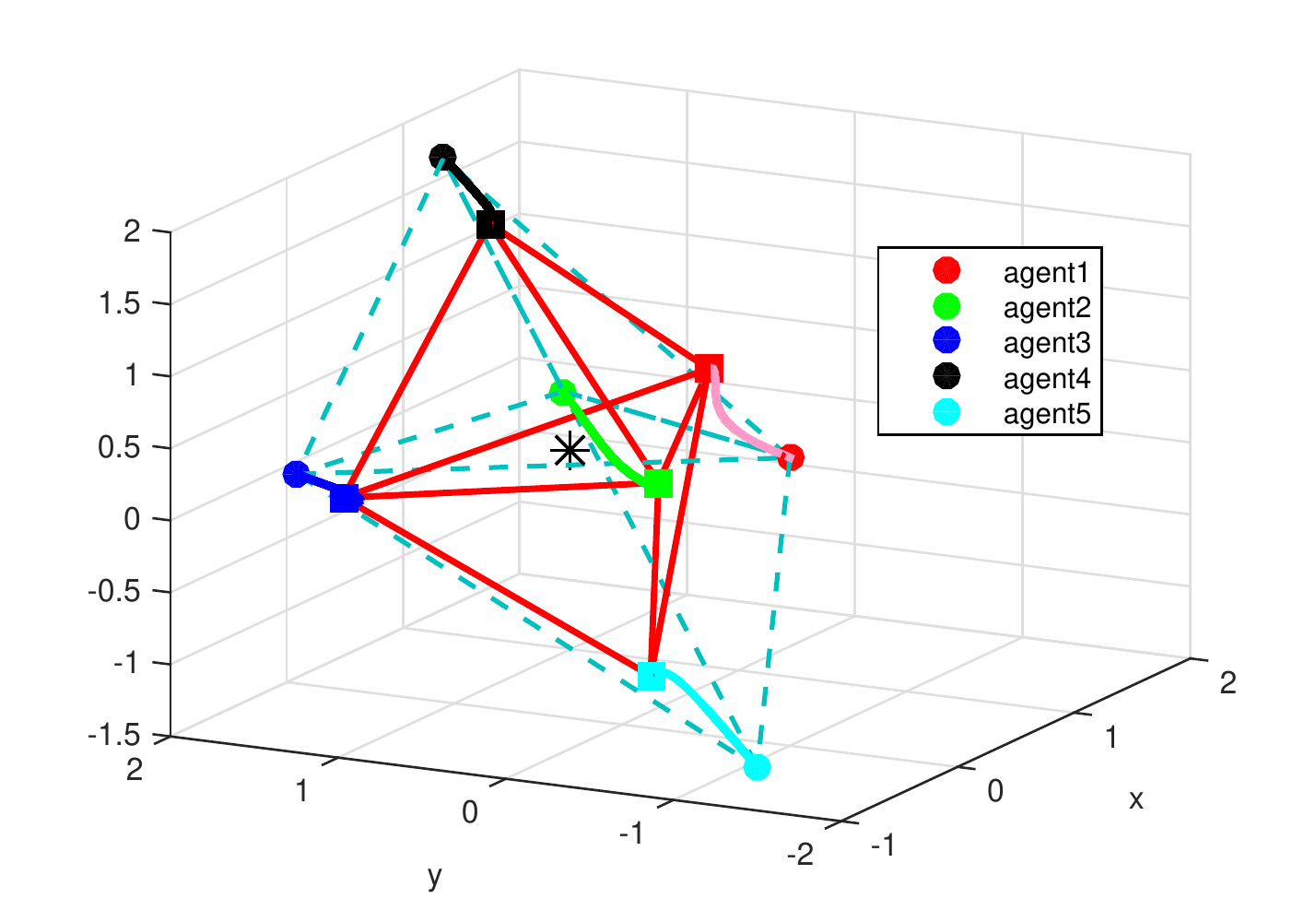}
  \caption{Simulation on stabilization control of a double tetrahedron formation in 3-D space with centralized event controller. The initial and final positions are denoted by circles and squares, respectively. The initial formation is denoted by dashed lines, and the final formation is denoted by red solid lines. The black star denotes the formation centroid, which is stationary.  }
  \label{shape_new}
  \label{fig:env}
\end{figure}
\begin{figure}
  \centering
  \includegraphics[width=60mm]{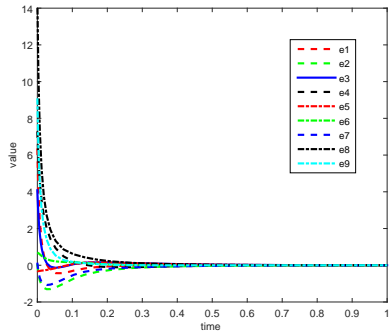}
  \caption{Exponential convergence of the distance errors with centralized event controller. }
  \label{distance_error}
  \label{fig:env}
\end{figure}
\begin{figure}
  \centering
  \includegraphics[width=70mm]{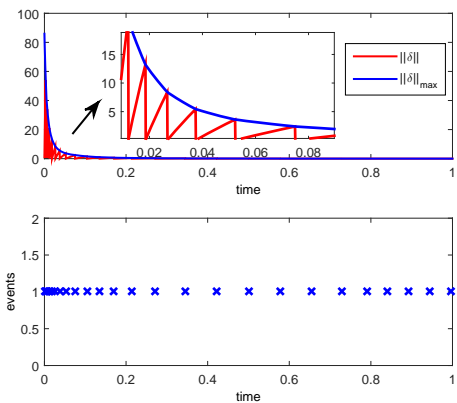}
  \caption{Performance of the centralized event-based controller.
  Top: evolution of $\|\delta\|$ and $\|\delta\|_{\max} = \gamma \|R(t)^T e(t)\|$.  Bottom: event triggering instants.}
  \label{shape_new}
  \label{fig:env}
\end{figure}

\begin{figure}
  \centering
  \includegraphics[width=75mm]{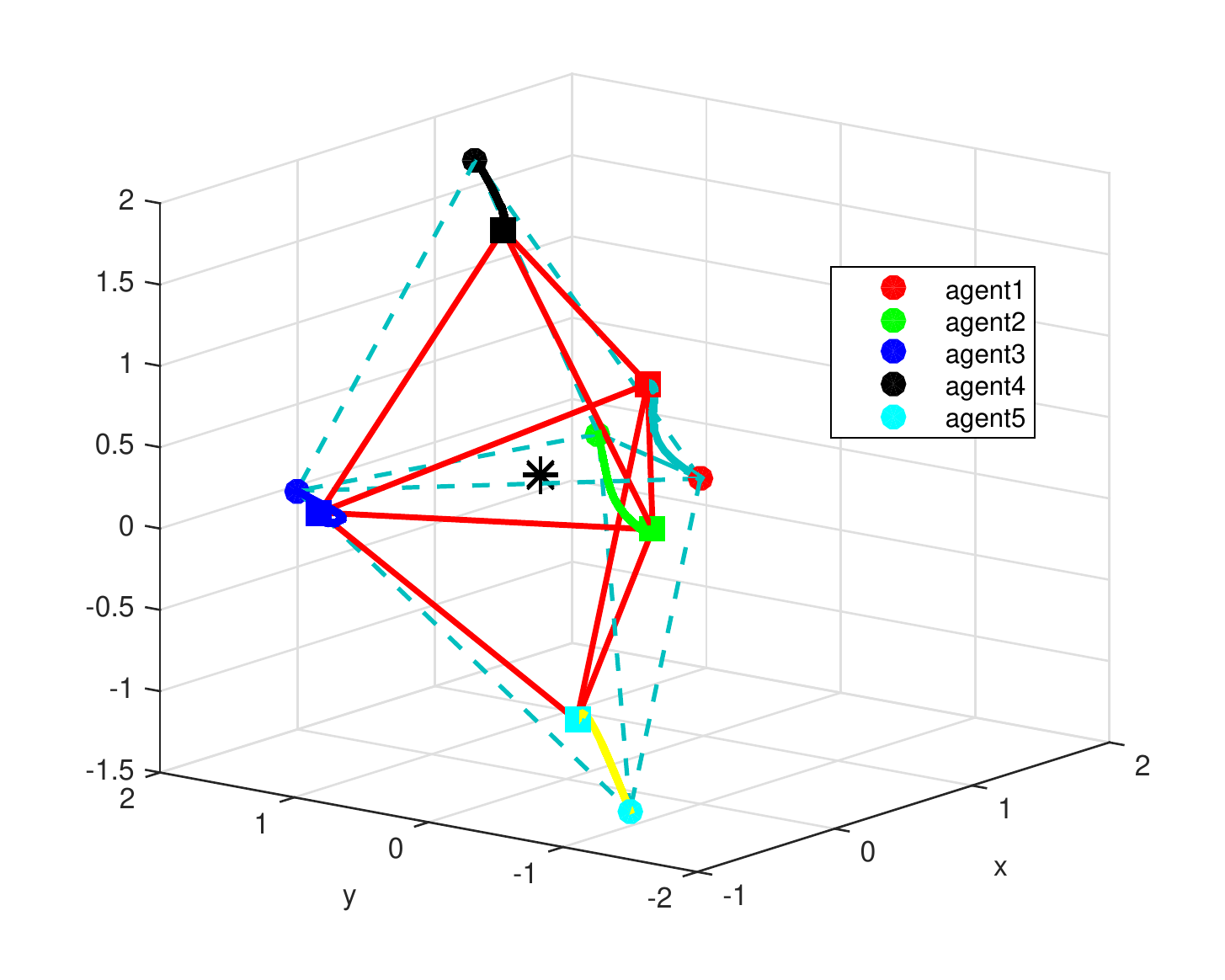}
  \caption{Simulation on stabilization control of a double tetrahedron formation in 3-D space with  distributed event controller. The initial and final positions are denoted by circles and squares, respectively. The initial formation is denoted by dashed lines, and the final formation is denoted by red solid lines. The black star denotes the formation centroid, which is \emph{not} stationary. }
\end{figure}
\begin{figure*}[t]
  \centering
  \includegraphics[width=130mm]{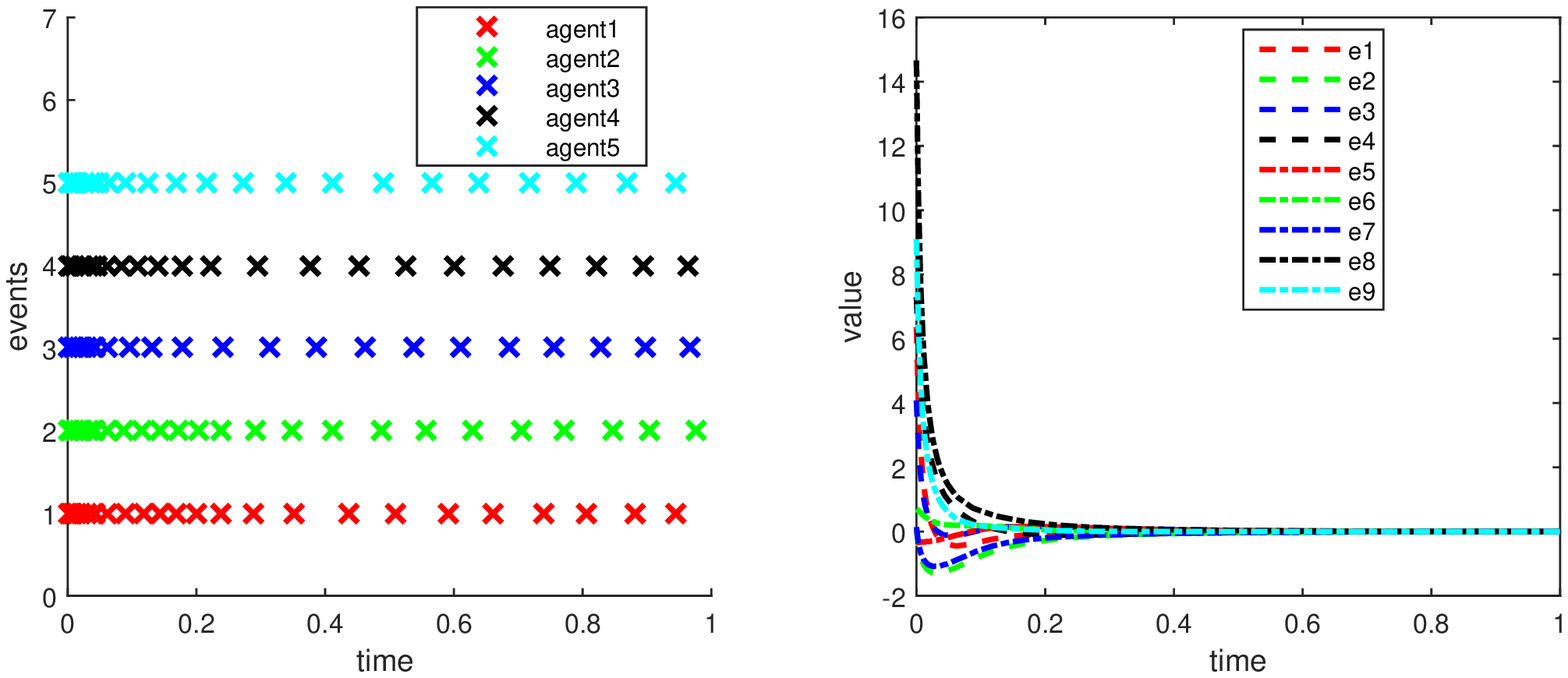}
  \caption{Controller performance of the distributed event-based formation system \eqref{position_system_distributed} with distributed event function \eqref{distributed_trigger_local_function}. Left:  Event triggering instants for each agent. Right: Exponential convergence of the distance errors with distributed event controller. }
\end{figure*}

We then perform another simulation on stabilizing the same formation shape by applying the proposed distributed event-based controller  \eqref{position_system_distributed} and distributed event function \eqref{distributed_trigger_local_function}. Agents' initial positions are set the same as for the simulation with the centralized event controller. The parameters $\gamma_{i}, i=1,2,\ldots,5$ are set as $0.8$ and $a_{i}, i=1,2,\ldots,5$ are set as $0.6$. The trajectories of each agent, together with the initial shape and final shape are depicted in Fig. 5. The   event times for each agent and the exponential convergence of each distance error are depicted in Fig. 6. Note that no $\|\{R(t)^T e(t)\}_i\|$ crosses zero values at any finite time and Zeno-triggering is excluded. Furthermore,   the distance error system shown in Fig. 6 also demonstrates almost the same convergence property as shown in Fig. 3.


Lastly, we  show simulations with the same formation shape  by using the modified event triggering function \eqref{distributed_modified_event}. The exponential decay term is chosen as $v_i \text{exp}(-\theta_i t) = \text{exp}(-10 t)$ with $v_i = 1$ and $\theta_i =10$ for each agent.  Fig.~7 shows   event-triggering  times for each agent   as well as the  convergence of each distance error. As can be observed from  Fig. 7,   Zeno-triggering is strictly excluded with the modified event function \eqref{distributed_modified_event}, while the convergence of the distance error system behaves almost the same to Fig. 3 and Fig.~6.  It should be noted that in comparison with the controller performance and simulation examples discussed in \cite{sun2015event, Liu2015event, CCC_CHENFEI}, the proposed event-based rigid formation controllers in this paper demonstrate equal or even better performance, while complicated controllers and unnecessary assumptions in  \cite{sun2015event, Liu2015event, CCC_CHENFEI} are avoided.
\begin{figure*}
  \centering
  \includegraphics[width=130mm]{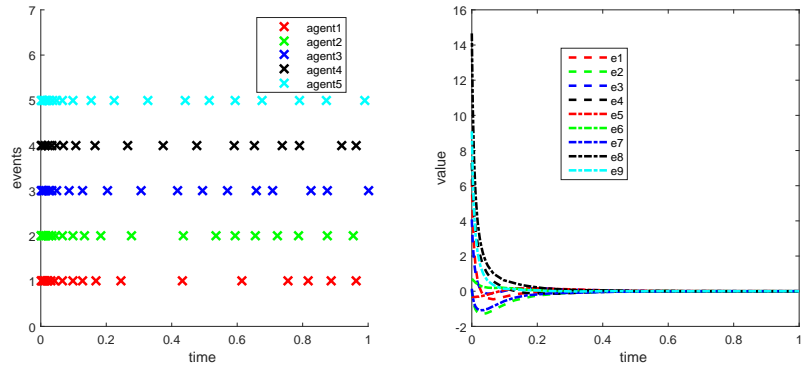}
  \caption{Controller performance of the distributed event-based formation system \eqref{position_system_distributed} (event function \eqref{distributed_modified_event} with an exponential decay term). Left:  Event triggering instants for each agent. Right: Exponential convergence of the distance errors with distributed event controller. }
\end{figure*}








\section{Conclusions}  \label{sec:conclusions}
In this paper we have discussed in detail the design of feasible event-based controllers to stabilize rigid formation shapes. A centralized event-based formation control is proposed first, which guarantees the exponential convergence of   distance errors  and also excludes the existence of Zeno triggering. Due to a careful design of the triggering error and event function, the controllers are much simpler and require  much less computation/measurement resources, compared with the results reported in   \cite{CCC_CHENFEI, sun2015event, Liu2015event }. We then further propose a distributed event-based controller such that each agent can  trigger a local event to update its control input based on only local measurement. The event  feasibility and triggering behavior have been discussed in detail, which also guarantees   Zeno-free behavior  for the event-based formation system and exponential convergence of the distance error system. A modified distributed event function is proposed, by which the Zeno triggering is strictly excluded for each individual agent.  A  future topic is to explore possible extensions of the current results on single-integrator models to the double-integrator rigid formation system \cite{TAC_COMBINED} with event-based control strategy to enable velocity consensus and rigid flocking behavior.

\section*{Acknowledgment}
This work was   supported by the National Science Foundation
of China under Grants 61761136005, 61703130 and 61673026,  and also supported by the Australian Research Council's Discovery Project DP-160104500 and DP-190100887.



%
\bibliography{Event_Formation}

\begin{thebibliography}{10}
\providecommand{\url}[1]{#1}
\csname url@samestyle\endcsname
\providecommand{\newblock}{\relax}
\providecommand{\bibinfo}[2]{#2}
\providecommand{\BIBentrySTDinterwordspacing}{\spaceskip=0pt\relax}
\providecommand{\BIBentryALTinterwordstretchfactor}{4}
\providecommand{\BIBentryALTinterwordspacing}{\spaceskip=\fontdimen2\font plus
\BIBentryALTinterwordstretchfactor\fontdimen3\font minus
  \fontdimen4\font\relax}
\providecommand{\BIBforeignlanguage}[2]{{%
\expandafter\ifx\csname l@#1\endcsname\relax
\typeout{** WARNING: IEEEtran.bst: No hyphenation pattern has been}%
\typeout{** loaded for the language `#1'. Using the pattern for}%
\typeout{** the default language instead.}%
\else
\language=\csname l@#1\endcsname
\fi
#2}}
\providecommand{\BIBdecl}{\relax}
\BIBdecl

\bibitem{oh2012survey}
K.-K. Oh, M.-C. Park, and H.-S. Ahn, ``A survey of multi-agent formation
  control,'' \emph{Automatica}, vol.~53, pp. 424--440, 2015.

\bibitem{anderson2008rigid}
B.~D.~O. Anderson, C.~Yu, B.~Fidan, and J.~M. Hendrickx, ``Rigid graph control
  architectures for autonomous formations,'' \emph{IEEE Control Systems
  Magazine}, vol.~28, no.~6, pp. 48--63, 2008.

\bibitem{krick2009stabilisation}
L.~Krick, M.~E. Broucke, and B.~A. Francis, ``Stabilisation of infinitesimally
  rigid formations of multi-robot networks,'' \emph{International Journal of
  Control}, vol.~82, no.~3, pp. 423--439, 2009.

\bibitem{meng2016formation}
Z.~Meng, B.~D.~O. Anderson, and S.~Hirche, ``Formation control with mismatched
  compasses,'' \emph{Automatica}, vol.~69, pp. 232--241, 2016.

\bibitem{meng20183}
------, ``On {3-D} formation control with mismatched coordinates,'' \emph{IEEE
  Transactions on Control of Network Systems}, vol.~5, no.~3, pp. 1492--1502,
  2018.

\bibitem{dorfler2010geometric}
F.~Dorfler and B.~Francis, ``Geometric analysis of the formation problem for
  autonomous robots,'' \emph{IEEE Transactions on Automatic Control}, vol.~55,
  no.~10, pp. 2379--2384, 2010.

\bibitem{oh2014distance}
K.-K. Oh and H.-S. Ahn, ``Distance-based undirected formations of
  single-integrator and double-integrator modeled agents in n-dimensional
  space,'' \emph{International Journal of Robust and Nonlinear Control},
  vol.~24, no.~12, pp. 1809--1820, 2014.

\bibitem{anderson2014counting}
B.~D.~O. Anderson and U.~Helmke, ``Counting critical formations on a line,''
  \emph{SIAM Journal on Control and Optimization}, vol.~52, no.~1, pp.
  219--242, 2014.

\bibitem{astrom2008event}
K.~J. Astr{\"o}m, ``Event based control,'' in \emph{Analysis and design of
  nonlinear control systems}.\hskip 1em plus 0.5em minus 0.4em\relax Springer,
  2008, pp. 127--147.

\bibitem{heemels2012introduction}
W.~Heemels, K.~H. Johansson, and P.~Tabuada, ``An introduction to
  event-triggered and self-triggered control.'' in \emph{Proc. of the 51st
  Conference on Decision and Control}, 2012, pp. 3270--3285.

\bibitem{sun2018event}
Z.~Sun, N.~Huang, B.~D.~O. Anderson, and Z.~Duan, ``{Event-based multiagent
  consensus control: Zeno-free triggering via Lp signals},'' \emph{IEEE
  Transactions on Cybernetics, article in press}, pp. 1--13, 2018.

\bibitem{ZhongkuiTNNLS}
B.~Cheng and Z.~Li, ``Designing fully distributed adaptive event-triggered
  controllers for networked linear systems with matched uncertainties,''
  \emph{IEEE Transactions on Neural Networks and Learning Systems}, pp. 1--11,
  2018.

\bibitem{zhang2018cooperative}
Y.~Zhang, H.~Li, J.~Sun, and W.~He, ``Cooperative adaptive event-triggered
  control for multiagent systems with actuator failures,'' \emph{IEEE
  Transactions on Systems, Man, and Cybernetics: Systems}, pp. 1--11, 2018.

\bibitem{Zhongkui2018TAC}
B.~Cheng and Z.~Li, ``Fully distributed event-triggered protocols for linear
  multi-agent networks,'' \emph{IEEE Transactions on Automatic Control}, pp.
  1--8, 2018.

\bibitem{seyboth2012event}
G.~S. Seyboth, D.~V. Dimarogonas, and K.~H. Johansson, ``Event-based
  broadcasting for multi-agent average consensus,'' \emph{Automatica}, vol.~49,
  no.~1, pp. 245--252, 2012.

\bibitem{Duan2018}
K.~Liu, P.~Duan, Z.~Duan, H.~Cai, and J.~Lu, ``Leader-following consensus of
  multi-agent systems with switching networks and event-triggered control,''
  \emph{IEEE Transactions on Circuits and Systems I: Regular Papers}, vol.~65,
  no.~5, pp. 1696--1706, May 2018.

\bibitem{qin2018leader}
J.~Qin, W.~Fu, Y.~Shi, H.~Gao, and Y.~Kang, ``Leader-following practical
  cluster synchronization for networks of generic linear systems: an
  event-based approach,'' \emph{IEEE Transactions on Neural Networks and
  Learning Systems}, pp. 1--10, 2018.

\bibitem{jiang2018synchronization}
C.~Jiang, H.~Du, W.~Zhu, L.~Yin, X.~Jin, and G.~Wen, ``Synchronization of
  nonlinear networked agents under event-triggered control,'' \emph{Information
  Sciences}, vol. 459, pp. 317--326, 2018.

\bibitem{Qingchen2018}
Q.~Liu, M.~Ye, J.~Qin, and C.~Yu, ``Event-triggered algorithms for
  leader-follower consensus of networked {Euler-Lagrange} agents,'' \emph{IEEE
  Transactions on Systems, Man, and Cybernetics: Systems}, pp. 1--13, 2018.

\bibitem{zhao2018edge}
Y.~Zhao, Y.~Liu, G.~Wen, W.~Ren, and G.~Chen, ``Edge-based finite-time protocol
  analysis with final consensus value and settling time estimations,''
  \emph{IEEE Transactions on Cybernetics}, pp. 1--10, 2018.

\bibitem{duan2018event}
P.~Duan, K.~Liu, N.~Huang, and Z.~Duan, ``Event-based distributed tracking
  control for second-order multiagent systems with switching networks,''
  \emph{IEEE Transactions on Systems, Man, and Cybernetics: Systems, article in
  press}, pp. 1--11, 2018.

\bibitem{ge2018survey}
X.~Ge, Q.-L. Han, D.~Ding, X.-M. Zhang, and B.~Ning, ``A survey on recent
  advances in distributed sampled-data cooperative control of multi-agent
  systems,'' \emph{Neurocomputing}, vol. 275, pp. 1684--1701, 2018.

\bibitem{zhang2017overview}
X.-M. Zhang, Q.-L. Han, and B.-L. Zhang, ``An overview and deep investigation
  on sampled-data-based event-triggered control and filtering for networked
  systems,'' \emph{IEEE Transactions on Industrial Informatics}, vol.~13,
  no.~1, pp. 4--16, 2017.

\bibitem{PENG2018113}
C.~Peng and F.~Li, ``A survey on recent advances in event-triggered
  communication and control,'' \emph{Information Sciences}, vol. 457-458, pp.
  113 -- 125, 2018.

\bibitem{nowzari2017event}
C.~Nowzari, E.~Garcia, and J.~Cortes, ``Event-triggered communication and
  control of network systems for multi-agent consensus,'' \emph{Automatica, to
  appear}, pp. 1--36, 2018.

\bibitem{li2018event}
X.~Li, X.~Dong, Q.~Li, and Z.~Ren, ``Event-triggered time-varying formation
  control for general linear multi-agent systems,'' \emph{Journal of the
  Franklin Institute, article in press}, pp. 1--17, 2018.

\bibitem{ge2017distributed}
X.~Ge and Q.-L. Han, ``Distributed formation control of networked multi-agent
  systems using a dynamic event-triggered communication mechanism,'' \emph{IEEE
  Transactions on Industrial Electronics}, vol.~64, no.~10, pp. 8118--8127,
  2017.

\bibitem{viel2017distributed}
C.~Viel, S.~Bertrand, M.~Kieffer, and H.~Piet-Lahanier, ``Distributed
  event-triggered control for multi-agent formation stabilization and
  tracking,'' \emph{arXiv preprint arXiv:1709.06652}, 2017.

\bibitem{yi2016formation}
X.~Yi, J.~Wei, D.~V. Dimarogonas, and K.~H. Johansson, ``Formation control for
  multi-agent systems with connectivity preservation and event-triggered
  controllers,'' \emph{arXiv preprint arXiv:1611.03105}, 2017.

\bibitem{Cortes_team_triggered}
C.~Nowzari and J.~Cortes, ``Team-triggered coordination for real-time control
  of networked cyber-physical systems,'' \emph{IEEE Transactions on Automatic
  Control}, vol.~61, no.~1, pp. 34--47, Jan 2016.

\bibitem{sun2015event}
Z.~Sun, Q.~Liu, C.~Yu, and B.~D.~O. Anderson, ``{Generalized controllers for
  rigid formation stabilization with application to event-based controller
  design},'' in \emph{Proc. of the 2015 European Control Conference
  (ECC'15)}.\hskip 1em plus 0.5em minus 0.4em\relax IEEE, 2015, pp. 217--222.

\bibitem{Liu2015event}
Q.~Liu, Z.~Sun, J.~Qin, and C.~Yu, ``{Distance-based formation shape
  stabilisation via event-triggered control},'' in \emph{Proc. of the 34th
  Chinese Control Conference(CCC'15)}.\hskip 1em plus 0.5em minus 0.4em\relax
  IEEE, 2015, pp. 6948--6953.

\bibitem{CCC_CHENFEI}
L.~Bai, F.~Chen, and W.~Lan, ``Decentralized event-triggered control for rigid
  formation tracking,'' in \emph{Proc. of the 34th Chinese Control Conference
  (CCC'15)}.\hskip 1em plus 0.5em minus 0.4em\relax IEEE, 2015, pp. 1262--1267.

\bibitem{SMA14TACsub}
S.~Mou, A.~S. Morse, M.~A. Belabbas, Z.~Sun, and B.~D.~O. Anderson,
  ``Undirected rigid formations are problematic,'' \emph{IEEE Transactions on
  Automatic Control}, vol.~61, no.~10, pp. 2821--2836, 2016.

\bibitem{sun2017rigid}
Z.~Sun, S.~Mou, B.~D.~O. Anderson, and A.~S. Morse, ``Rigid motions of {3-D}
  undirected formations with mismatch between desired distances,'' \emph{IEEE
  Transactions on Automatic Control}, vol.~62, no.~8, pp. 4151--4158, 2017.

\bibitem{horn2012matrix}
R.~A. Horn and C.~R. Johnson, \emph{Matrix analysis}.\hskip 1em plus 0.5em
  minus 0.4em\relax {Cambridge University Press}, 2012.

\bibitem{asimow1979rigidity}
L.~Asimow and B.~Roth, ``{The rigidity of graphs, II},'' \emph{Journal of
  Mathematical Analysis and Applications}, vol.~68, no.~1, pp. 171--190, 1979.

\bibitem{hendrickson1992conditions}
B.~Hendrickson, ``Conditions for unique graph realizations,'' \emph{SIAM
  Journal on Computing}, vol.~21, no.~1, pp. 65--84, 1992.

\bibitem{sun2016exponential}
Z.~Sun, S.~Mou, B.~D. Anderson, and M.~Cao, ``Exponential stability for
  formation control systems with generalized controllers: A unified approach,''
  \emph{Systems \& Control Letters}, vol.~93, pp. 50--57, 2016.

\bibitem{cortes2008discontinuous}
J.~Cortes, ``Discontinuous dynamical systems,'' \emph{IEEE Control Systems},
  vol.~28, no.~3, 2008.

\bibitem{sun2014finite}
Z.~Sun, S.~Mou, M.~Deghat, and B.~D.~O. Anderson, ``Finite time distributed
  distance-constrained shape stabilization and flocking control for
  d-dimensional undirected rigid formations,'' \emph{International Journal of
  Robust and Nonlinear Control}, vol.~26, no.~13, pp. 2824--2844, 2016.

\bibitem{SE_N_INVARIANCE}
C.-I. Vasile, M.~Schwager, and C.~Belta, ``Translational and rotational
  invariance in networked dynamical systems,'' \emph{IEEE Transactions on
  Control of Network Systems}, vol.~5, no.~3, pp. 822--832, 2018.

\bibitem{dorfler_thesis}
F.~Dorfler, ``Geometric analysis of the formation control problem for
  autonomous robots,'' Diploma thesis, University of Toronto. Supervisor: Prof.
  Bruce Francis, 2008.

\bibitem{zhiyong_CDC15}
Z.~Sun, U.~Helmke, and B.~D.~O. Anderson., ``Rigid formation shape control in
  general dimensions: an invariance principle and open problems,'' in
  \emph{Proc. of the 54th Conference on Decision and Control}.\hskip 1em plus
  0.5em minus 0.4em\relax IEEE, 2015, pp. 6095--6100.

\bibitem{park2018distance}
M.-C. Park, Z.~Sun, B.~D.~O. Anderson, and H.-S. Ahn, ``Distance-based control
  of {$K_n$} formations in general space with almost global convergence,''
  \emph{IEEE Transactions on Automatic Control}, vol.~63, no.~8, pp.
  2678--2685, 2018.

\bibitem{ramazani2017rigidity}
S.~Ramazani, R.~Selmic, and M.~de~Queiroz, ``Rigidity-based multiagent layered
  formation control,'' \emph{IEEE Transactions on Cybernetics}, vol.~47, no.~8,
  pp. 1902--1913, 2017.

\bibitem{Raik_SIAM}
R.~Suttner and Z.~Sun, ``Formation shape control based on distance measurements
  using {Lie} bracket approximations,'' \emph{SIAM Journal on Control and
  Optimization}, vol.~56, no.~6, pp. 4405--4433, 2018.

\bibitem{ames2006stability}
A.~D. Ames, P.~Tabuada, and S.~Sastry, ``{On the stability of Zeno
  equilibria},'' in \emph{International Workshop on Hybrid Systems: Computation
  and Control}.\hskip 1em plus 0.5em minus 0.4em\relax Springer, 2006, pp.
  34--48.

\bibitem{tabuada2007event}
P.~Tabuada, ``Event-triggered real-time scheduling of stabilizing control
  tasks,'' \emph{IEEE Transactions on Automatic Control}, vol.~52, no.~9, pp.
  1680--1685, 2007.

\bibitem{fan2013distributed}
Y.~Fan, G.~Feng, Y.~Wang, and C.~Song, ``Distributed event-triggered control of
  multi-agent systems with combinational measurements,'' \emph{Automatica},
  vol.~49, no.~2, pp. 671--675, 2013.

\bibitem{dimarogonas2012distributed}
D.~V. Dimarogonas, E.~Frazzoli, and K.~H. Johansson, ``Distributed
  event-triggered control for multi-agent systems,'' \emph{IEEE Transactions on
  Automatic Control}, vol.~57, no.~5, pp. 1291--1297, 2012.

\bibitem{SUN2018264}
Z.~Sun, N.~Huang, B.~D.~O. Anderson, and Z.~Duan, ``Comments on `distributed
  event-triggered control of multi-agent systems with combinational
  measurements','' \emph{Automatica}, vol.~92, pp. 264 -- 265, 2018.

\bibitem{khalil1996nonlinear}
H.~K. Khalil, \emph{Nonlinear systems}.\hskip 1em plus 0.5em minus 0.4em\relax
  Prentice hall New Jersey, 1996.

\bibitem{TAC_COMBINED}
M.~Deghat, B.~D.~O. Anderson, and Z.~Lin, ``Combined flocking and
  distance-based shape control of multi-agent formations,'' \emph{IEEE
  Transactions on Automatic Control}, vol.~61, no.~7, pp. 1824--1837, 2016.

\end{thebibliography}
\bibliographystyle{IEEEtran}

\end{document}